\documentclass[runningheads]{llncs}
\usepackage[T1]{fontenc}
\usepackage{float}
\usepackage{graphicx}
\usepackage{mathtools}
\usepackage{amsmath}
\usepackage{enumitem}
\usepackage{mdframed}
\usepackage{caption}
\usepackage{xcolor}

\usepackage [
    n, % or lambda
    advantage,
    operators,
    sets,
    adversary,
    landau,
    probability,
    notions,
    logic,
    ff,
    mm,
    primitives,
    events,
    complexity,
    oracles,
    asymptotics,
    keys
]{cryptocode}

\providecommand{\coloneq}{\coloneqq}

\newcommand{\calA}{\mathcal{A}}

\newcommand{\calE}{\mathcal{E}}

\newcommand{\calP}{\mathcal{P}}

\newcommand{\calT}{\mathcal{T}}

\newtheorem{assumption}{Assumption}

\begin{document}
\title{Sealed-bid Auctions on Blockchain with Timed Commitment Outsourcing}
% Uncomment each \orcidID below and replace the placeholder with the author's ORCID.
\author{Jichen Li\inst{1}\orcidID{0000-0002-8711-5040}
\and
Yuanchen Brian Tang\inst{2}%\orcidID{0000-0000-0000-0000}
\and
Jing Chen\inst{1}\thanks{Corresponding author.}\orcidID{0009-0002-6992-4909}
\and
Xiaotie Deng\inst{3}\orcidID{0000-0002-5282-6467}
\and
Bin Xiao\inst{4}\orcidID{0000-0003-4223-8220}
}
\authorrunning{J. Li et al.}
\institute{Department of Computer Science and Technology, Tsinghua University\\
\url{\{jichenli, jchencs\}@mail.tsinghua.edu.cn} \and
University of Illinois at Urbana-Champaign, United States\\
\url{yctang1@illinois.edu} \and
City University of Hong Kong, Kowloon, Hong Kong\\
\url{csdeng@cityu.edu.hk} \and
Hong Kong Polytechnic University, Kowloon, Hong Kong\\
\url{csbxiao@connect.polyu.hk}}
\maketitle              % typeset the header of the contribution
\begin{abstract}
Sealed-bid auctions play a crucial role in blockchain ecosystems. Previous works introduced viable blockchain sealed-bid auction protocols, leveraging timed commitments for bid encryption. However, a crucial challenge remains unresolved in these works: Who should bear the cost of decrypting these timed commitments?

This work introduces a timed commitment outsourcing market as a solution to this challenge. We first introduce an aggregation scheme for timed commitments, which combines all bidders' timed commitments into one while ensuring security and correctness and allowing a varying number of bidders. Next, we remodel the utility of auctioneers and timed commitment solvers, developing a new timed commitment competition mechanism and combining it with the sealed-bid auction to form a two-sided market. The protocol includes bid commitment collection, timed commitment solving, and payment. Through game-theoretical analysis, we prove that our protocol satisfies Dominant Strategy Incentive Compatibility (DSIC) for bidders and Bayesian Incentive Compatibility (BIC) for solvers, and achieves optimal revenue for the auctioneer among a large class of mechanisms. Finally, in the forced-opening-only model in which bidders cannot self-open, we prove that no mechanism can achieve positive expected revenue for the auctioneer while satisfying DSIC and Individual Rationality (IR) for both bidders and solvers, together with ex-post IR for the auctioneer.

\end{abstract}
\section{Introduction}

Sealed-bid auction, where all bidders simultaneously submit their confidential bids to the auctioneer, is a cornerstone in the realm of auctions.
In the most widely known variant of the second-price sealed-bid auction, also known as the Vickrey auction~\cite{vickrey1961counterspeculation}, 
the auctioneer sells the item to the highest bidder at the second-highest price, and it is optimal for bidders to report their values truthfully. However, a fundamental premise of sealed-bid auctions is that the auctioneer must be a trusted party, or the mechanism must be credible, a condition that has been proven to be unattainable through mechanism design alone~\cite{akbarpour2020credible}. 
This has impelled research into leveraging cryptography and blockchain to mitigate auctioneer misbehaviors~\cite{kosba2016hawk,blass2018strain,li2021blockchain,david2022fast}.

However, in these protocols, the inherent transparency of blockchain necessitates that users first submit encrypted commitments, which are then revealed after all bids have been collected. 
This approach, however, introduces two critical vulnerabilities: (1) a bidder might maliciously refuse to reveal their commitment after observing others' bids; and (2) miners could maliciously prevent bidders' commitment-revealing transactions from being included in the blockchain (similar attacks have been studied in the payment-channel network (PCN) domain~\cite{wadhwa2022he}). To simultaneously address both issues, Tyagi et al.~\cite{tyagi2023riggs} innovatively employed timed commitments combined with efficient range proofs to construct sealed-bid auctions on the blockchain. 
Recent systems further study zero-deposit delayed execution, private auctions at scale, rational liveness, censorship resistance, and TEE-backed confidential compute~\cite{zhang2024zeroauction,novakovic2026cryptobazaar,huang2025optimizing,alpos2026censorship,gebele2025sealed}.
Yet, their study left open a critical problem:

\emph{Who bears the cost of computing these timed commitments?}

The significance of this question stems from the fact that the decryption costs associated with the commitment mechanism profoundly influence the utility of both bidders and auctioneers, potentially altering the auction's mechanism and bidding strategies, and also critically impacting the overall security of the entire auction protocol. 
%Parallel studies of decentralized verifiable computation and blockchain procurement likewise treat solver competition, deadlines, liveness risk, and incentive compatibility as central mechanism-design questions rather than implementation details~\cite{gong2024v3rified,bahrani2026adversarial,garimidi2026beyond}.
For instance, consider a scenario with a single solver for timed commitments, where each commitment invocation incurs an individual cost $c$. 
If bidders are required to bear a portion of $c$, then running the Vickrey auction may no longer satisfy Individual Rationality (IR), as the second price plus the decryption cost might exceed the first price. 
Conversely, if the auctioneer bears a portion of the individual cost $c$, a single bidder can collude with the solver to extract the auctioneer's payment without actually solving the commitment.
% However, if we aggregate timed commitments into one, the solver can only extract payment by colluding with all bidders, which is much harder.
Consequently, the design of a robust timed commitment solving mechanism is paramount for decentralized sealed-bid auctions on the blockchain, necessitating a meticulous examination of participant incentives.

Building on top of~\cite{tyagi2023riggs}, the contribution of this paper includes:
\begin{itemize}
    \item[1.] We generalize the timed commitment scheme from Malavolta et al.~\cite{malavolta2019homomorphic} and propose an aggregatable timed commitment scheme. Compared to the scheme by Glaeser et al.~\cite{glaeser2023cicada}, which is also aggregatable, our scheme supports auction scenarios with varying numbers of bidders with only a one-time setup.
    
    \item[2.] We consider auctioneers who account for both the decryption cost and the decryption time and remodel the timed commitment sealed-bid auctions as a two-sided game by introducing a commitment outsourcing competition for solvers. We redesign the utility functions of the auctioneer and solver, drawing on Bayesian auction theory to design a decryption bidding mechanism.

    \item[3.] We propose solutions that satisfy the Bayesian Incentive Compatibility (BIC) and Bayesian Individual Rationality (BIR) properties for the solver while the bidder's bidding strategies remain unaffected.
    This ensures that our decryption mechanism supports any sealed-bid auction protocol without compromising the auction's properties or the protocol's security.

    \item[4.] Through game-theoretical analysis, we prove that our mechanism achieves maximum auctioneer revenue in a large class of mechanisms, and its optimality is also shown by a complementing impossibility result in achieving DSIC and IR for both bidders and solvers.
\end{itemize}

\section{Model}
The formal definitions of time-lock puzzles and their properties (correctness, security, homomorphism, aggregatability, and non-malleability) are recalled in Appendix~\ref{appendix:prelim}. We use the terms \emph{timed commitment} and \emph{time-lock puzzle} interchangeably.

\subsection{Our Model for Timed Commitment Auction}
We consider a single-item auction scenario. 
%where multiple auctions can be initiated simultaneously.
Our timed commitment auction model involves three types of participants: \emph{bidder}, \emph{auctioneer}, and \emph{solver}.

\smallskip
\noindent\textbf{Bidder.} A bidder is an entity participating in the auction, aiming to obtain the auctioned item by placing bids. In an auction, suppose there are $m$ bidders, and each bidder $i \in [m]$ has a private value $v_i$ for the auctioned item. 
A bidder should design a bidding strategy to determine her bid $b_i$ according to her private value.

\smallskip
\noindent\textbf{Auctioneer.} The auctioneer organizes the auction, aiming to sell the item and maximize revenue. However, unlike traditional auctions, where the result is immediately available after receiving all bidders' bids, in a timed commitment auction, the auctioneer can only obtain the result and revenue after the commitments are revealed. Therefore, in our model, the auctioneer aims to maximize \emph{revenue per unit time}, as defined in Definition~\ref{def:auctioneer}.

\smallskip
\noindent\textbf{Solver.} Our model includes $n$ solvers, numbered as $\{1, \dots, n\}$, responsible for providing computational power to decrypt the commitment and receive payment.
In practice, a solver may use rented generic compute or specialized hardware. The platform can benchmark a machine with hourly price $\theta$ and measured speed $s(\theta)$ modular squarings per second, then assign marginal cost $c(\theta)=\frac{\theta}{3600\,s(\theta)}$ per squaring.
The public curve $\kappa$ is fit from these speed--price pairs, with $\kappa(c(\theta))=s(\theta)$. We assume that $\kappa$ is strictly increasing, so that achieving a higher speed requires accepting a higher marginal cost per squaring; this trade-off is what makes the solver's choice of $c_j$ strategic.
For example, if a benchmark reports $s=10^8$ squarings per second at $\theta=\$0.32$ per hour, then $c\approx 8.9\times 10^{-13}$ dollars per squaring and a puzzle with $\mathcal{T}=10^{13}$ costs about $\$8.9$ to open. This calculation is illustrative; deployment should use measured prices for the actual machines.

In each auction, a solver chooses one point on the public curve, but its effective cost may differ because of electricity prices, idle capacity, bulk-rental discounts, or depreciation. Solver $j$ therefore has a private discount ability $l_j\in[1,\bar{l}]$, drawn independently from a common prior distribution $F$, allowing it to use hardware of speed $\kappa(c_j)$ at marginal cost $c_j/l_j$.
Here, $\bar{l}$ is an arbitrarily large upper bound used for convenience in the proof.
Thus, each solver $j$ should strategically choose its computation cost $c_j$ in order to maximize its profit.

A complete timed commitment auction mechanism requires the auctioneer to determine a tuple $<\calA, \calP^{b},\calP^{s}>$ which are the \emph{Allocation Rule}, \emph{Bidder Payment Rule}, and \emph{Solver Payment Rule}, defined as follows:

\begin{definition}[Timed Commitment Auction] 
	A timed commitment auction is defined by determining a tuple $<\calA, \calP^{b}, \calP^{s}>$, where:
	% \begin{itemize}
 
	\noindent\textbf{Allocation Rule} $\calA: \mathbb{R}^m \rightarrow \{0,1\}^m$ maps bidders' reported bids $\vec{b} = (b_1, \dots, b_m)$ to an $m$-dimensional vector $(a_1(\vec{b}), \dots, a_m(\vec{b}))$ satisfying $\sum_{i=1}^m a_i(\vec{b}) = 1$, where $a_i(\vec{b})$ is the probability of the item allocated to bidder $i$. 
        % Given bidders' reported bids $\vec{b} = (b_1, \dots, b_m)$, an allocation rule $\calA: R^m \rightarrow \{0,1\}^m$ maps $\vec{b}$ to an $m$-dimensional vector $(a_1(\vec{b}), \dots, a_m(\vec{b}))$ satisfying $\sum_{i=1}^m a_i(\vec{b}) = 1$, where $a_i(\vec{b})$ is the amount of the item allocated to bidder $i$.
	
	\noindent\textbf{Bidder Payment Rule} $\calP^{b}: \mathbb{R}^m \rightarrow \mathbb{R}^m$ maps bidders' bid $\vec{b}$ to an $m$-dimensional vector $(p^b_1(\vec{b}), \dots, p^b_m(\vec{b}))$, where $p^b_i(\vec{b})$ is the amount of payment that bidder $i$ should pay to the auctioneer.
				
	\noindent\textbf{Solver Payment Rule} $\calP^{s}: \mathbb{R}^{m+n} \rightarrow \mathbb{R}^n$ maps bids $\vec{b}$ and solvers' computation costs $\vec{c} = (c_1, \dots, c_n)$ to an $n$-dimensional vector $(p^s_1(\vec{b}, \vec{c}), \dots, p^s_n(\vec{b}, \vec{c}))$, where $p^s_j(\vec{b}, \vec{c})$ is the amount of payment made by the auctioneer to solver $j$.
        %Given solvers' computation costs $\vec{c} = (c_1, \dots, c_n)$, a Solver Payment Rule $\calP^{s}: R^{2n} \rightarrow R^n$ maps $(\vec{b}, \vec{c})$ to an $n$-dimensional vector $(p^s_1(\vec{b}, \vec{c}), \dots, p^s_n(\vec{b}, \vec{c}))$, where $p^s_j(\vec{b}, \vec{c})$ is the amount of payment to solver $j$.
	% \end{itemize}
\end{definition}

We now define the utility function for each type of participant.

\begin{definition}[Bidder Utility]
\label{def:bidder}
	Given the allocation rule $\calA$ and bidder payment rule $\calP^b$, the utility of bidder $i \in [m]$ under bids $\vec{b}$ and private value $v_i$ is defined as
	\begin{align*}
		u^b_i(\vec{b}) = a_i(\vec{b}) \cdot v_i - p^b_i(\vec{b}).
	\end{align*}
\end{definition}

\begin{definition}[Solver Utility]
\label{def:solver}
	Given the solver payment rule $\calP^s$, the utility of solver $j \in [n]$ under the puzzle complexity parameter $\calT$, bids $\vec{b}$, computation costs $\vec{c}$ and discount ability $l_j$ is defined as
	\begin{align*}
		u^s_j(\vec{b}, \vec{c}) = p^s_j(\vec{b}, \vec{c}) - \frac{\calT c_j}{l_j}.
	\end{align*}
\end{definition}

The auctioneer receives the payment from bidders and pays solvers for commitment revealing. 
We consider the auctioneer's revenue per unit of time.

\begin{definition}[Auctioneer Revenue]
\label{def:auctioneer}
	Given the bidder payment rule $\calP^b$ and the solver payment rule $\calP^s$, the revenue of the auctioneer under bids $\vec{b}$, computation cost $\vec{c}$, cost-speed function $\kappa$, and time parameter $\calT$ is defined as
	\begin{align*}
		R = \left(\sum_{i=1}^m p^b_i(\vec{b}) - \sum_{j=1}^n p^s_j(\vec{b}, \vec{c})\right) \cdot \frac{\kappa(\max{\vec{c}})}{\calT}.
	\end{align*}
\end{definition}

\subsection{Incentive Properties}
We use the standard notions of dominant-strategy incentive compatibility (DSIC), individual rationality (IR), Bayes-Nash equilibrium, Bayesian incentive compatibility (BIC), and Bayesian IR for a common prior; see~\cite{myerson1981optimal,roughgarden2010algorithmic}. DSIC means truthful participation is optimal for every fixed profile of others' actions, and IR means the prescribed action yields nonnegative utility. In the Bayesian case, the comparison is in expectation conditional on the participant's type.

We analyze the economic game after accepting the timed-privacy guarantee of the puzzle. A polynomial-time solver obtains no advantage from trying to break the puzzle; the detailed argument is in Appendix~\ref{appendix:honest-solver}. Thus the strategic choice is the computation cost and submission time, not an attack on the cryptographic primitive.

\section{Blockchain Sealed-bid Auction Protocol}
\label{sec:pro}

\begin{figure}[!ht]
	\begin{mdframed}
		\underline{$PSetup(1^\lambda, \calT, L)$}
		\begin{enumerate}[label=(\arabic*)]
			\item Sample a pair of prime numbers $(p, q)$ such that $p = 2 p' + 1$ and $q = 2 q' + 1$, where $p'$ and $q'$ are also prime, and set $N \coloneq p \cdot q$.
			\item Uniformly sample $g' \sample \ZZ_N^*$ and set $g = - (g')^2 (\text{mod } N)$.
			\item Compute $h \coloneq g^{2^{\calT}} (\text{mod } N^{L+1})$, where $L$ is a system capacity parameter; the exponent can be reduced modulo the order of $g$ in $\mathbb{Z}_{N^{L+1}}^*$, which the setup computes from $\varphi(N)$.
			\item Output $pp \coloneq (\calT, N, L, g, h)$.
		\end{enumerate}
		
		\underline{$PGen(pp, M, k \in [M], e \in \{0,\dots,N-1\})$}
		\begin{enumerate}[label=(\arabic*)]
			\item Parse $pp \coloneq (\calT, N, L, g, h)$, where $M \le L$ is the capacity of the current auction.
			\item Uniformly sample $r \sample \{1, \cdots, N^2\}$.
			\item Generate elements $u \coloneq g^r (\text{mod } N^{M+1})$ and $w \coloneq h^{r \cdot N} \cdot (1 + N^k)^e (\text{mod } N^{M+1})$, with $h$ reduced modulo $N^{M+1}$.
			\item Output $Z_k \coloneq (u, w)$ as the timed commitment.
		\end{enumerate}
		
		\underline{$PAggr(pp, M, Z_1, ... , Z_m)$}
		\begin{enumerate}[label=(\arabic*)]
			\item Parse $pp \coloneq (\calT, N, L, g, h)$.
			\item Parse each puzzle $Z_i \coloneq (u_i, w_i) \in \mathbb{Z}^*_{N^{M+1}} \times \mathbb{Z}^*_{N^{M+1}}$, with $u_i \bmod N \in \mathbb{J}_{N}$.
			\item Compute $U = \Pi_{i=1}^m u_i (\text{mod } N^{M+1})$ and $W = \Pi_{i=1}^m w_i (\text{mod } N^{M+1})$.
			\item Output the puzzle $Z' = (U, W)$.
		\end{enumerate}
		
		\underline{$PASolve(pp, M, Z, m)$}
		\begin{enumerate}[label=(\arabic*)]
			\item Parse $pp \coloneq (\calT, N, L, g, h)$.
			\item Parse the aggregate puzzle $Z \coloneq (U, W)$.
			\item Compute $X \coloneq U^{2^\calT} \pmod{N^{M+1}}$ by repeated squaring.
			\item Compute $e_1 \coloneq \frac{W/(X^N) \pmod{N^2} - 1}{N}$.
			\item Given $e_1, \dots, e_{k-1}$, compute
			\[
				e_k \coloneq
				\frac{W/(X^N \cdot \prod_{i=1}^{k-1}(1+N^i)^{e_i}) \pmod{N^{k+1}} - 1}{N^k}
			\]
			for every $k \in [m]$.
			\item Output $(e_1, ..., e_m)$.
		\end{enumerate}
	\end{mdframed}
	\caption{Aggregatable Timed Commitment Scheme.}
	\label{fig:AHTLP}
\end{figure}

In this section, we introduce our timed commitment auction protocol on the blockchain, including the construction of aggregatable homomorphic time-lock puzzles, the bidding protocol for bidders, and the commitment unveiling protocol for solvers.

\subsection{Aggregatable Homomorphic Time-Lock Puzzle}
\label{sec:ATLP}

Our aggregatable time-lock puzzle is based on the homomorphic time-lock puzzle construction introduced in~\cite{malavolta2019homomorphic}, and the means of achieving aggregatability were also briefly discussed in~\cite{thyagarajan2020verifiable}. The security of this construction relies on the complexity assumption of squaring operations in a finite field, as introduced by~\cite{rivest1996timelock}.

Let $N = p \cdot q$ be an RSA integer, where $p$ and $q$ are large primes.
%A trivial time-lock puzzle for a secret $s$ with respect to time $T$ can consist of a quintuple $(N, T, x, x^{2 ^ T} \cdot k, \text{Enc}(k, e))$, where $(x, k)$ are elements uniformly sampled from $\mathbb{Z}_N^*$.
%Then, in order to make the puzzle homomorphic,
Malavolta et al.~\cite{malavolta2019homomorphic} construct a linear homomorphic timed commitment as $(N, T, x^r, (x^{N \cdot 2^T})^r \cdot (1 + N)^e)$, where $r$ is uniformly sampled from $\{1, \cdots, N^2\}$, and its distribution modulo $\varphi(N)$ is statistically close to uniform over $\{1, \cdots, \varphi(N)\}$.
This construction follows the idea that the group $\mathbb{Z}_{N^2}^*$ can be represented as a product of the group generated by $(1 + N)$ and the $N$-th residue group $\{x^N : x \in \mathbb{Z}_N^*\}$. 
Our scheme additionally works in the subgroup $\mathbb{J}_N \subseteq \mathbb{Z}_N^*$ of elements whose Jacobi symbol is $+1$, so that the Jacobi symbol leaks no information about the exponent.

However, this construction does not yet satisfy the property of aggregatability (formally defined in Appendix~\ref{appendix:prelim}). We further extend it, observing that when $N$ is a strong RSA integer, every element of the subgroup $1 + N\mathbb{Z}_{N^{m+1}} \subseteq \mathbb{Z}_{N^{m+1}}^*$ admits a unique coordinate decomposition, as stated in the following Lemma, the proof of which is in Appendix~\ref{proof:lemma-group}.

% However, in the auction model, our goal is not to perform additive or multiplicative operations between time-lock puzzles but to combine multiple puzzles into a single puzzle that can be handed over to the server for computation. Inspired by Malavolta et al.~\cite{malavolta2019homomorphic}, we observe that when $N$ is a strong RSA integer, the group $\mathbb{Z}_{N^{m+1}}^*$ can be decomposed according to the following Lemma:

\begin{lemma}
\label{lemma:group}
	Let $N$ be a strong RSA integer and $m \ge 1$. Every element $y$ of the subgroup $1 + N\mathbb{Z}_{N^{m+1}} = \{y \in \mathbb{Z}_{N^{m+1}}^* : y \equiv 1 \pmod{N}\}$ admits a unique decomposition
	\[
		y = \prod_{k=1}^{m} (1 + N^k)^{t_k}, \qquad t_k \in \{0, 1, \dots, N-1\}.
	\]
	Moreover, $1 + N\mathbb{Z}_{N^{m+1}}$ is cyclic of order $N^m$, generated by $1 + N$.
\end{lemma}

To prevent an attacker from gaining information through the Jacobi symbol, we also replace the group $\mathbb{Z}_N^*$ with the subgroup $\mathbb{J}_N$ when constructing the encryption.
We also note that the original algorithm $PSolve$ is a special case of algorithm $PASolve$ with $m = 1$.
The scheme is shown in Fig.~\ref{fig:AHTLP}.

\smallskip
\noindent\textbf{Working level of the construction.}
An element that is known only modulo $N$ does not determine its $N$-th power modulo $N^{k+1}$ for $k \ge 2$: for instance, $1 + N \equiv 1 \pmod{N}$, yet $(1+N)^{N} \equiv 1 + N^{2} \not\equiv 1 \pmod{N^{3}}$.
All values that enter computations modulo $N^{k+1}$ must therefore be carried explicitly in $\mathbb{Z}_{N^{M+1}}^*$, where $M$ is the per-auction capacity announced at auction initiation: the setup publishes $h = g^{2^{\mathcal{T}}}$ reduced modulo $N^{L+1}$ for a system capacity $L$, bidders commit modulo $N^{M+1}$ with their auction's $M \le L$, and the solver's repeated-squaring chain runs modulo $N^{M+1}$ (still $\mathcal{T}$ sequential squarings, on a larger modulus).
With this convention, $X = U^{2^{\mathcal{T}}} = h^{R}$ holds as an exact identity in $\mathbb{Z}_{N^{M+1}}^*$ (where $h$ is reduced modulo $N^{M+1}$ and $R = \sum_i r_i$), and the extraction equations of $PASolve$ are well-defined.

\smallskip
\noindent\textbf{Correctness and Security.}
The correctness of this construction follows from the exact identities in $\mathbb{Z}_{N^{M+1}}^*$ together with the coordinate uniqueness of Lemma~\ref{lemma:group}, as verified in Appendix~\ref{appendix:ATLP}, and the security of our variant follows from the same hybrid sequence as the proof of Theorem 1 in~\cite{malavolta2019homomorphic}.
Crucially, our scheme ensures that as long as at least one bidder honestly constructs the timed commitment, their bid remains confidential even if the solver, all other bidders, and the auctioneer collude. This property effectively prevents potential bribery issues.

\smallskip
\noindent\textbf{Range proofs for aggregatable time-lock puzzles.}
Aggregatable time-lock puzzles are compatible with the range proof method presented in~\cite{tyagi2023riggs}; the details are discussed in Appendix~\ref{appendix:comparison}.

% linearly on DDH\cite{castagnos2015linearly}
% \cite{wahby2020airdrop}

\smallskip
\noindent\textbf{Comparing with other possible solutions.}
A comparison with the packing protocol in Cicada~\cite{glaeser2023cicada} is provided in Appendix~\ref{appendix:comparison}.

% For instance, if the maximum value of each bid is $B$ and there are $m$ bidders, Cicada's scheme requires each commitment to be constructed over the group $\mathbb{Z}_{B^{2m}}^*$ to prevent bid overflow. 
% In contrast, our scheme only requires each commitment to be constructed over the group $\mathbb{Z}_{B^{m+1}}^*$, offering good efficiency. 
% This design choice allows for dynamic participation of bidders, as the group size consideration depends on the individual bid size rather than the number of participants.
% For instance, if the maximum value of each bid is $B$, Cicada's scheme 需要事先知道bidder数量m，并选择一个group size $N \geq B^m$ to prevent bid overflow， 因此构造的每一个commitment over the group $\mathbb{Z}_{N^2}^*$。In contract, our scheme 只需要choose a group size $N \geq B$, 无需事先知道bidder的数量，同时index i的bidder构造的commitment 只需 over the group $\mathbb{Z}_{N^{i+1}}^*$。对同样的第m个bidder，cicada的commitment group size要大于B^{2m}，而我们commitment使用的group size只需大于B^{m+1}，因此我们的scheme效率更高。

%Another advantage of our scheme is that it maintains additive homomorphism under the same $k$-order, allowing a group of people to pool their money to buy a product together without revealing their bids before the auction finishes.
%(suit for ~\cite{})
% 他们也是那个论文里discuss的。当然我们指出这两种方案可以一起结合使用。
%如果group N大小一样，那么我们commitment的bit更多，但是solver解题on group N，所以没有影响
%如果bid range same, then both of scheme should use a group of N^m, but cicada should predeter the group of numbers.

\subsection{Timed Commitment Auction on Bidder Side}

\begin{figure}[!ht]
    \begin{mdframed}
        \textbf{Initialization:} The auction is initialized with public parameters $(\calT, N, L, g, h) \leftarrow PSetup(1^{\lambda}, \calT, L)$.

    \underline{Phase 1: Auction Initiation}
        \begin{enumerate}[label=(\arabic*)]
            \item The auctioneer publishes an auction smart contract on the blockchain at $t_0$.
            \item The auctioneer then publishes the auction capacity $M$ and the bid cap $B$, together with $t, t_{\text{Open}}$, for the deadlines of the bid collection phase and bid self-open phase.
        \end{enumerate}
    \underline{Phase 2: Bid Collection}
        \begin{enumerate}[label=(\arabic*)]
            \item Each bidder $i$ chooses a bidding price $b_i$ according to her value $v_i$.
            \item To place a bid, a bidder must:
            \begin{enumerate}[label=(\alph*)]
                \item Send a transaction to the auction smart contract to identify her index $i$.
                \item Generate commitment $(u_i, w_i) \leftarrow PGen(pp, M, i, b)$.
                \item Submit puzzle $Z_i=(u_i,w_i)$ with a NIZK $\pi_i$ proving knowledge of $(b_i,r_i)$ such that $Z_i=PGen(pp,M,i,b_i;r_i)$ and $b_i$ lies in the published bid range.
            \end{enumerate}
            \item The auctioneer ends bid collection at time \(t_0 + t\).
        \end{enumerate}
    \underline{Phase 3: Bid Opening}
        \begin{enumerate}[label=(\arabic*)]
            \item The smart contract aggregates the accepted puzzles to $(U, W) \leftarrow PAggr(pp, M, Z_1, \dots, Z_m)$.
            \item The smart contract sets $Z = (U, W)$ and begins the timed commitment competition in Fig.~\ref{fig:sol}, in which the aggregated commitment is force-opened by a solver.
            \item While a solver is opening the aggregated commitment, bidder $i$ may self-open by sending $\pi^i_{\mathrm{Open}}=(b_i,r_i)$, where $r_i$ is the randomness used in $PGen$, before time $t_0+t+t_{\mathrm{Open}}$.
        \end{enumerate}
        
        % \begin{enumerate}[label=(\arabic*)]
        %     \item Users provide openings \(\pi_{\text{Open}}\) computed in phase 1 to reveal bid \(b\).
        %     \item Auctioneer VERIFY \(\pi_{\text{Open}}\)
        %     \item At time \(t_0 + t + t_{\text{Open}}\), ends.
        % \end{enumerate}
        
    \underline{Phase 4: Payment Execution}
        \begin{enumerate}[label=(\arabic*)]
            \item The smart contract obtains revealed bids $\vec{b}$ from the accepted $PASolve(pp,M,Z,m)$ output.
            \item It computes the vector of solvers' inferred costs $\vec{c}$ from submission times.
            \item It allocates the item by $\calA(\vec{b})$, collects bidders' payments by $\calP^b(\vec{b})$, and pays solvers by $\calP^s(\vec{b},\vec{c})$.
            \item The smart contract closes the auction.
        \end{enumerate}
        \textbf{Output:} Item allocation, and payment for bidders and solvers.
    \end{mdframed}
    \caption{Blockchain sealed-bid auction protocol for selling one item, using an aggregatable timed commitment to ensure the credibility of bids, and solver competition to ensure that all bids can be decrypted on time.}
    \label{fig:auction}
\end{figure}

Given a timed commitment auction mechanism $<\calA, \calP^b, \calP^s>$ and an aggregatable timed commitment scheme $(PSetup, PGen, PAggr, PASolve)$, the blockchain auction protocol consists of five phases (detailed in Figure~\ref{fig:auction}):

\smallskip
\noindent\textbf{Phase 0: System setup}.
A trusted party runs $PSetup(1^{\lambda}, \calT, L)$ with a security parameter $\lambda$ and a hardness parameter $\calT$ to generate public parameters $pp$, fixing a system capacity $L$.
This system setup phase can be executed once and supports multiple auctions; each auction announces its own capacity $M \le L$.
After publishing $pp$, the setup party must securely erase $p$, $q$, $p'$, $q'$, $\varphi(N)$, the randomness used to sample $g'$, and all intermediate setup state.

\smallskip
\noindent\textbf{Phase 1: Auction initiation\footnote{At this stage, the auctioneer only needs to choose the type of timed commitment auction (e.g., whether the bidder payment rule and solver payment rule follow a first-price or second-price rule) and set the appropriate parameters (such as reserve price, auction end time, etc.). The rest of the protocol can then be executed through a smart contract.}}. 
The auctioneer introduces the item, initiates an auction smart contract, and sets up duration times $t, t_{\text{Open}}$ and a bidder capacity $M$, indicating the deadlines for bid collection and bid self-open (block numbers represent time in the blockchain). We denote the auction start time as $t_0$.

\noindent\textbf{Phase 2: Bid collection}. 
During this phase, each bidder starts by sending a transaction to the smart contract to identify her unique index $i$, then chooses a bidding price $b_i$ and submits a timed commitment as $Z_i \leftarrow PGen(pp,M,i,b_i)$.
This phase will end at time $t_0 + t$.

\noindent\textbf{Phase 3: Bid opening}.
The auctioneer aggregates time commitments from bidders by $Z \leftarrow PAggr(pp, M, Z_1, ..., Z_m)$ and starts a timed commitment solving competition (Fig.~\ref{fig:sol}).
Bidders can also self-open their sealed bids between time $t_0 + t$ and $t_0 + t + t_{\text{Open}}$, which is a fallback opening path.
The contract pays the solver reward only to the first valid aggregate solution. Partial self-openings do not stop the aggregate competition; they provide public fallback data if no valid aggregate solution arrives before $t_0 + t + t_d$. If all bidders validly self-open before any aggregate solution is accepted, the contract closes the competition and pays no solver reward.

\noindent\textbf{Phase 4: Payment execution}.
In this phase, the auctioneer retrieves all open bids from the timed commitment solving competition and the solvers' inferred costs $\vec{c}$.
The smart contract then computes the auction result, allocates the item by $\calA(\vec{b})$, collects bidders' payments by $\calP^b(\vec{b})$, and pays solvers according to $\calP^s(\vec{b}, \vec{c})$.

% \begin{itemize}
%     \item 
%     Initially, the Auctioneer introduces the item, sets the public parameters for the time-lock puzzle, and submits a commitment for a reserve price. 
%     Following this, the first slot, known as the bidding slot, commences and lasts for $t$ duration. 
%     During this phase, each bidder can submit a sealed bid commitment using the public parameters. 
%     Additionally, each server submits a sealed bid detailing their cost and speed.
%     \item
%     The second slot has a duration of $t_2$ time. In this slot, the Auctioneer is responsible for disclosing the valid bids, and the winning server is tasked with solving the puzzle and receiving payment.
%     \item 
%     The third slot extends for a duration of $T / d_i$, during which the server resolves the puzzle. Subsequently, the Auctioneer reveals the solution and declares the auction winner.
% \end{itemize}

\subsection{Timed Commitment Solving Competition}

\begin{figure}[!ht]
    \begin{mdframed}
        \textbf{Initialization:} The smart contract publishes the public parameters $pp$, aggregated timed commitment $Z$ and bidder number $m$, and starts the timed commitment competition.
    
    \underline{Phase 1: Due Time Setup}
        \begin{enumerate}[label=(\arabic*)]
            \item The smart contract publishes a due time $t_d$, a duration measured from the start of the solving competition ($t_0 + t$) and chosen to be at least $t_{\mathrm{Open}}$, according to the solver's discount ability distribution $F$, the cost-speed function $\kappa$, and time parameter $\calT$.
            The way to calculate the optimal due time which can maximize the auctioneer's expected revenue is shown in
            Theorem~\ref{thm:hbw} (Eq.~\eqref{formula:due-time}).
            \item The smart contract publishes the solver payment rule $\calP^s$, due time $t_d$, and fraud-report reward $\rho$, and requires each solver to post a deposit $D^s\ge B/2+\rho$, where $B$ is the published bid cap; by Theorem~\ref{thm:hbw} the total solver payment never exceeds $B/2$, so the deposit covers any payment plus the fraud reward $\rho$.
        \end{enumerate}

    \underline{Phase 2: Solver Competition}
        \begin{enumerate}[label=(\arabic*)]
            \item Each solver $j$ selects computation cost $c_j$, computes $X=U^{2^\calT}\pmod{N^{M+1}}$ and $\vec{b}=PASolve(pp,M,Z,m)$, and attaches a succinct proof $\pi_j^{\mathrm{sol}}$ that $X$ was obtained by the prescribed repeated squaring.
            \item Solver $j$ posts the deposit $D^s$ and submits $(X,\vec{b},\pi_j^{\mathrm{sol}})$ at time $t_0+t+\frac{\calT}{\kappa(c_j)}$.
        \end{enumerate}

    \underline{Phase 3: Payment Allocation}
        \begin{enumerate}[label=(\arabic*)]
            \item The smart contract verifies $\pi_j^{\mathrm{sol}}$ and the coordinate-extraction equations; it accepts the first valid solution and rejects or slashes an invalid submission.
            \item The smart contract allocates the item and executes payments according to $<\calA,\calP^b,\calP^s>$ once a valid solution is accepted.
            \item If no valid aggregate solution is submitted before $t_0 + t + t_d$, the auction is decided only from valid self-opened bids.
        \end{enumerate}
    \end{mdframed}

    \caption{Timed commitment competition protocol where solvers compete on how fast they solve the timed commitment to win the payment from the auctioneer.}
    \label{fig:sol}
\end{figure}

When the auctioneer announces the aggregate timed commitment and due time, the solver will begin a timed commitment solving competition. 
Notably, the competition's rewards are initially unknown to the solver; the exact value will only be revealed once the bids are solved. 
Therefore, at this stage, each solver must select their computation cost based on their discount ability and their estimates of other participants' types. 
Within the payment-proportional, rank-based class, under \(n\)-maximum-payment-regular solver types and a linear \(\kappa\), Section~\ref{sec:analysis} shows that allocating all payments to the first solver maximizes auctioneer revenue. Therefore, the payment allocation phase of the protocol commences immediately after the first solver submits the result.
The detailed protocol is shown in Figure~\ref{fig:sol}.

% In our protocol(detailed in Figure \ref{fig:protocol}), as we have discussed, users need to use our Scalar time-lock puzzle to submit their bids in order to ensure every bid is opened finally. The public parameters of is set in the initialization phase. In the first phase, when a participant commits a bid, they compute a puzzle $(u,v)$ hiding their bids and send it to the auctioneer.

% In the second phase, users can open their puzzles themselves by providing the bid $b$ and a proof $\pi_{Open}$. After receiving the openings, the auctioneer VERITY the openings.

% Those unopened bids will be opened collectively in the subsequent force-opening phase. As we have discussed, before giving these bids to the computing parties, the auctioneer will aggregate them into a large time-lock puzzle, avoiding extra issues about separate allocation of these bids. Then, Auctioneer evaluates privately the computing resources, choose to activate a suitable external protocol. A solution is obtained afterwards and auctioneer can compute, publish the result of the auction.

\smallskip 
\noindent\textbf{Bidder Verification}
The bidder verification mechanism, which prevents solvers from fraudulently obtaining payment with false results, is described in Appendix~\ref{appendix:verification}; the public verification of aggregates and solver outputs is described in Appendix~\ref{appendix:public-verification}.

\section{Revenue-Optimal Mechanisms}
\label{sec:analysis}

In this section, we analyze the timed commitment auction mechanism that maximizes revenue for the auctioneer. We find that \emph{in a large class of mechanisms with a natural structure, the optimal mechanism takes the form of a second price auction with a reserve price for the bidder and a winner-takes-all competition for the solver: see} \textbf{Theorem~\ref{thm:hbw}.}
Since this mechanism satisfies DSIC~\cite{vickrey1961counterspeculation} for the bidder, we will subsequently use $\vec{v}$ to replace the notation $\vec{b}$ and denote the payment to the winning solver as $V(\vec{v})$.
Concretely, the bidder-side rules $\calA$ and $\calP^b$ implement a second-price auction with a reserve price: the highest bidder wins if and only if its bid meets the reserve, and pays the maximum of the second-highest bid and the reserve.

% Let $v_{(k)}$ be the $k$-th highest value. $i*$ be bidder winner, $j*$ be solver winner.

% \begin{align*}
%     a_i(\vec{b}) = \begin{cases}
%     1, &b_{(1)} \geq b^r \\
%     0, &\text{ozws}
%     \end{cases},
%     p^b_i(\vec{b}) = \begin{cases}
%     \max{(b_{(2)}, r)}, &b_{(1)} \ge b^r \\
%     0, &\text{ozws}
%     \end{cases},
%     p^s_j(\vec{b}, \vec{c}) = \begin{cases}
%     1, &c_{(1)} \ge c^r \\
%     0, &\text{ozws}
%     \end{cases}.
% \end{align*}
% \subsection{Notations}
% Bidders private values are denoted by $\vec v$, in which $v_i \sim F^b$, a distribution in $[\underline v,\overline v]$, $\overline v$ can be replaced by $\infty$.
% Solvers' discount $\vec l$ are denoted by $\vec l$, in which $l_j \sim F^s$, a distribution in $[1, \overline l]$.

% Allocation rule is $\vec x^{\mathcal{A}}(\vec v)$
% The bidding strategy $\vec b: b_i:R^+\rightarrow R^+$
% The interim allocation rule $\vec x(v)=\{x_i(v_i)\}_{[n]}$ is defined by $x_i(v_i)=E[x_i^{\mathcal{A}}(\vec b(\vec v))|v_i]$. We'll omit unnecessary subscript $i$.

% It is easy to see that the auction for the bidder is IC, so in the following analysis, we only use symbol $\mathbf{v}$, as $\mathbf{b} = \mathbf{v}$.

\subsection{Equilibrium Analysis}

We first characterize the BNE of solvers in the proposed mechanism.
Building on the work of Chawla et al.~\cite{chawla2012crowdsourcing}, we prove that when the solver's discount ability distribution satisfies the n-maximum-payment regularity condition, an extension of Myerson's lemma demonstrates the existence of a unique symmetric BNE.

\begin{definition}[n-maximum-payment-regular]
\label{def:reg}
For a given distribution $F$ with density function $f$ and an integer $n$, we define the maximum payment virtual value, $\psi_n(z)$, as
$$
\psi_n(z)=zF(z)^{n-1}-\frac{1-F(z)^n}{nf(z)}.
$$
A distribution $F$ is said to be n-maximum-payment-regular if $\psi_n(\cdot)$ is a monotone non-decreasing function wherever it is non-negative.
%, that is, $\psi_n(v)>0$ implies $\psi'_n(v)\ge 0$. 
%The distribution is said to be maximum-payment regular if, for all positive integers $n$, $\psi_n(\cdot)$ is monotone non-decreasing wherever it is non-negative.
\end{definition}

This regularity condition covers standard smooth distributions used in auction theory rather than only pathological cases. For example, uniform distributions on a bounded interval and truncated exponential distributions on a bounded interval provide common instances; more generally, the condition holds whenever the distribution of the maximum solver discount ability is Myerson-regular.

\begin{theorem}
\label{thm:bne}
In a winner-takes-all competition for solvers in which a solution is eligible for payment only if it is submitted before the due time $t_d$ (Fig.~\ref{fig:sol}), when the distribution $F$ of each solver's discount ability is n-maximum-payment-regular, a unique symmetric BNE exists with strategy
$$
c(l) = \frac{\mathbb{E}_{\vec {v}}[V(\vec v)]}{\calT}(lx(l)-\int_{1}^{l} x(z)dz),
$$ 
where $l$ is the solver's discount ability,
$
x(l)=\begin{cases}F^{n-1}(l), &l\ge \psi_n^{-1}(0)\\0,&l<\psi_n^{-1}(0)\end{cases}
$ is the winning probability, $\calT$ is the timed-commitment difficulty parameter, $\vec{v}$ is the bidders' value profile, and $V(\vec{v})$ is the payment to the winning solver.
\end{theorem}

The detailed proof is shown in Appendix~\ref{proof:thm-bne}.
Theorem~\ref{thm:sym} additionally shows that this BNE is unique.
Moreover, the equilibrium is interim individually rational for solvers; see Proposition~\ref{prop:solver-bir} in the appendix.

\subsection{Mechanism Optimality}

We introduce some definitions of the solver payment rule, which help us prove BNE's uniqueness and design the optimal $V(\vec{v})$ function.
The formal definitions of payment-proportional and rank-based payment rules, together with an intuitive explanation, are given in Appendix~\ref{appendix:payment-defs}.

We now prove that the BNE in Theorem~\ref{thm:bne} is the unique equilibrium by applying Theorem 3.1 of Chawla et al.~\cite{chawla2013auctions}.
The proof is provided in Appendix~\ref{proof:thm-sym}.

\begin{theorem}
\label{thm:sym}
    When the solver payment rule is payment-proportional and rank-based, a unique BNE exists for solvers.
\end{theorem}

%%%%%%%%
% An interim allocation rule is a product of expected scaling factor and a basic interim predetermined reward: for agent $i$,
% $$
% x_i(v_i)=E[x_i^{\mathcal{A}}(\vec b(\vec v))|v_i]=E[V(\omega)x_{0i}^{\mathcal{A}}(\vec b(\vec v))|v_i]=E_{\omega}[V(\omega)]*E[x_{0i}^{\mathcal{A}}(\vec b(\vec v))|v_i]:=E[V]x_{0i}(v_i)
% $$

% The V-proportional rank-based \textbf{interim} allocation (with a reserve price) $x_i(v_i)=E[x_i^{\mathcal{A}}(\vec b(\vec v))|v_i]$ is finally symmetric as $x_i^{\mathcal{A}}(\cdot)$ and $b_i(\cdot)$ are all symmetric now.
%%%%%%%%

In the following theorem, we demonstrate how the auctioneer determines the optimal value function $V$ and the due time for optimal revenue.

\begin{theorem}\label{thm:hbw}
In mechanisms where the payment rule is rank-based and payment-proportional, the optimal competition mechanism for solvers is a winner-takes-all competition when the distribution $F$ of solver discount abilities is $n$-maximum-payment-regular and the cost-speed function $\kappa$ is linear.
Specifically, when $\kappa(c) = \beta \cdot c + \gamma$ with $\beta > 0$ ($\gamma \leq \mathbb{E}_{\vec v}[v_{(2)}] \cdot n \beta \cdot \mathbb{E}_{l \sim F}[\psi_n(l)\overline x(l)] / \calT$), the optimal payment function $V$ and the optimal due time in the competition are given by:
\begin{equation}
    V(\vec{v}) = \frac{v_{(2)}}{2} - \frac{\gamma \cdot \calT \cdot v_{(2)}}{2n \beta \mathbb{E}_{l \sim F}[\psi_n(l)\overline x(l)]\cdot \mathbb{E}_{\vec v}[v_{(2)}]}, \quad t_d = \frac{\mathcal{T}}{\kappa(c(\psi_n^{-1}(0)))},
    \label{formula:due-time}
\end{equation}
where $v_{(2)}$ is the second-highest value among the bidders (or the reserve price), and $\calT$ is the timed commitment difficulty factor, $
\overline x(l)=\begin{cases}F^{n-1}(l), &l\ge \psi_n^{-1}(0)\\0,&l<\psi_n^{-1}(0)\end{cases}
$ is the winning probability.

% When agent values are distributed i.i.d. from a distribution that is n-maximum-payment regular:if $g$ is linear and with positive slope(i.e. $g(x)=kx+b,k>0$), then the optimal(maximizing auctioneer's revenue) all-pay auction(within V-proportional rank-based allocation with a reserve price and BNE existence) is highest-bid-wins-all with a reserve price.
% if $g$ is convex, the revenue reaches an lower bound which is maximized by the mechanism.
\end{theorem}

The proof of this theorem builds upon the work of Chawla et al.~\cite{chawla2012crowdsourcing}, with the proof provided in Appendix~\ref{proof:thm-hbw}.
Additionally, we explore how to select an appropriate $V$ function to maximize the auctioneer's revenue in a winner-takes-all competition when the function $\kappa$ is non-linear.
This result can be easily generalized from the proof of Theorem~\ref{thm:hbw}, and is stated in Appendix~\ref{cor:nonlinear}.

Essentially, the property of payment-proportionality separates the influence of the bidder and solver in the mechanism.
Whether the auctioneer can achieve a strictly larger revenue with a non-payment-proportional payment rule is left as an open question for future study (see Appendix~\ref{appendix:payment-defs}).

\subsection{Impossibility of two-sided DSIC}

% \begin{definition}{Social welfare}
%     Social welfare is a measure of the allocation result, which is $$
%     \sum_i{x_ival_i}*speed
%     $$
%     . We get the optimal expected social welfare to allocate the item to the bidder with the highest private value and let the puzzle be computed in fastest way. That is, using all the possible gain to compute the solution. Thus the optimal welfare is $$
%     E_{\vec {val},\vec v}[\max_ival_i*g(\max_ival_i*\max_jv_j)]
%     $$
% \end{definition}

In the previous analysis, we only required that the solver satisfies the properties of Bayesian Incentive Compatibility (BIC) and Individual Rationality (BIR). 
A natural question arises: Is there a mechanism that satisfies DSIC and IR for both bidders and solvers? 
We prove in Theorem~\ref{thm:two-side-DSIC} that when the bid vector can be revealed only by a solver, no such mechanism can generate positive revenue for the auctioneer when the bidder's value can be zero. 

% How can a two-sides DSIC mechanism approximate the optimal welfare? The answer is barely impossible. If we use the ratio of expected welfare$\frac{welf(M_{opt})}{welf(M)}$ to represent the approximating factor, we prove that NO mechanism can achieve a constant ratio under a proper distribution $U[0,1]$. 
\begin{theorem}
\label{thm:two-side-DSIC}
    When bidder values can be zero, the cost-speed function satisfies $\kappa(0) = 0$, and the bid vector can only be revealed through the forced opening of the timed commitment, any mechanism that is DSIC and ex-post IR for both bidders and solvers, and additionally ex-post IR for the auctioneer, must have zero auctioneer revenue on every value profile.
\end{theorem}
The detailed proof is shown in Appendix~\ref{proof:thm-two-side}.
%Proof sketch: With negligible knowledge of bids before opening them, for arbitrarily small computational cost, it's possible that the earnings from bidders cannot cover the outcome.
%尝试构造一个one bidder，two solver的场景作为反例，用公式推导

%expost IR for auctioneer.

%\input{Sections/7-Evaluation}
% Related work is discussed in Appendix (2-Related.tex is input from Sections/Appendix.tex).

\section{Conclusion and Future Directions}
We studied blockchain sealed-bid auctions with a timed commitment outsourcing market.
We introduced an aggregatable timed commitment scheme that supports a varying number of bidders after a one-time setup.
Building on it, we designed a second-price auction with a reserve price for bidders and a winner-takes-all solving competition for solvers.
This mechanism satisfies DSIC and IR for bidders, BIR and BIC for solvers, and maximizes auctioneer revenue among payment-proportional, rank-based mechanisms when the solver types are \(n\)-maximum-payment-regular.
We also proved that, when bidder values can be zero, $\kappa(0) = 0$, and the bid vector can only be revealed through the forced opening of the timed commitment, performed by the solvers, no mechanism can satisfy DSIC and IR for both bidders and solvers, together with auctioneer IR, while guaranteeing positive auctioneer revenue.

% An intriguing open question is whether there exists a mechanism that is not payment-proportional and allows the auctioneer to secure even greater revenue.
% Another interesting question is
% whether there exists a mechanism which is BIC for both bidders and solvers and achieves a higher revenue for the auctioneer.
% Notice that the revenue difference between DSIC and BIC mechanisms is an important question for single-sided auctions, and here it is worth pursuing for our two-sided game as well.

Future directions include non-payment-proportional mechanisms, mechanisms that are BIC for both bidders and solvers, dynamic interactions with reputation and strategy evolution, and stronger forms of solver collusion and bribery.

\begin{credits}
\subsubsection{\ackname}
This work was supported by NSFC/RGC Joint Research Scheme under Grant 62261160391 and Grant N\_PolyU529/22, and Beijing Advanced Innovation Center for Future Blockchain and Privacy Computing.
\end{credits}

% BibTeX with Springer LNCS bibliography style 'splncs04'.
%
\bibliographystyle{splncs04}
\bibliography{ref}

@article{diffie1976new,
  author       = {Whitfield Diffie and Martin E. Hellman},
  title        = {New directions in cryptography},
  journal      = {IEEE Transactions on Information Theory},
  volume       = {22},
  number       = {6},
  pages        = {644--654},
  year         = {1976},
  doi          = {10.1109/TIT.1976.1055638}
}

@inproceedings{garay2013rational,
  author       = {Juan A. Garay and
                  Jonathan Katz and
                  Ueli Maurer and
                  Bj{\"{o}}rn Tackmann and
                  Vassilis Zikas},
  title        = {Rational Protocol Design: Cryptography against Incentive-Driven Adversaries},
  booktitle    = {54th Annual {IEEE} Symposium on Foundations of Computer Science, {FOCS}
                  2013, 26-29 October, 2013, Berkeley, CA, {USA}},
  pages        = {648--657},
  publisher    = {{IEEE} Computer Society},
  year         = {2013}
}

@inproceedings{abadi2025scalable,
  series       = {ASIA CCS '25},
  title        = {Scalable Time-Lock Puzzle},
  booktitle    = {Proceedings of the 20th ACM Asia Conference on Computer and Communications Security},
  publisher    = {ACM},
  author       = {Abadi, Aydin and Ristea, Dan and Grigor, Artem and Murdoch, Steven},
  year         = {2025},
  month        = Aug,
  pages        = {839--855},
  collection   = {ASIA CCS '25}
}

@misc{gong2024v3rified,
  author       = {Gong, Tiantian and Kate, Aniket and Psomas, Alexandros and Terzoglou, Athina},
  title        = {V3rified: Revelation vs Non-Revelation Mechanisms for Decentralized Verifiable Computation},
  publisher    = {arXiv},
  year         = {2024},
  copyright    = {Creative Commons Attribution Share Alike 4.0 International}
}

@inbook{chung2025foundations,
  title        = {Foundations of Platform-Assisted Auctions},
  ISBN         = {9783032018786},
  ISSN         = {1611-3349},
  booktitle    = {Advances in Cryptology -- CRYPTO 2025},
  publisher    = {Springer Nature Switzerland},
  author       = {Chung, Hao and Wu, Ke and Shi, Elaine},
  year         = {2025},
  pages        = {135--169}
}

@misc{bahrani2026adversarial,
  author       = {Bahrani, Maryam and Neuder, Michael and Weinberg, S. Matthew},
  title        = {Adversarial procurement in blockchains},
  publisher    = {arXiv},
  year         = {2026},
  copyright    = {Creative Commons Attribution 4.0 International}
}

@misc{garimidi2026beyond,
  author       = {Garimidi, Pranav and Neuder, Michael and Roughgarden, Tim},
  title        = {Beyond Winner-Take-All Procurement Auctions},
  publisher    = {arXiv},
  year         = {2026},
  copyright    = {arXiv.org perpetual, non-exclusive license}
}

@inproceedings{xin2025check,
  title        = {``Check-Before-you-Solve'': Verifiable Time-Lock Puzzles},
  booktitle    = {2025 IEEE Symposium on Security and Privacy (SP)},
  publisher    = {IEEE},
  author       = {Xin, Jiajun and Papadopoulos, Dimitrios},
  year         = {2025},
  month        = May,
  pages        = {2037--2056}
}

@article{wu2024practical,
  title        = {Practical Verifiable Time-Lock Puzzle: Pre- and Post-Solution Verification},
  volume       = {16},
  ISSN         = {2073-8994},
  number       = {10},
  journal      = {Symmetry},
  publisher    = {MDPI AG},
  author       = {Wu, Zheyi and Liu, Haolin and Wang, Lei},
  year         = {2024},
  month        = Oct,
  pages        = {1347}
}

@inbook{zhang2024zeroauction,
  title        = {ZeroAuction: Zero-Deposit Sealed-bid Auction via Delayed Execution},
  ISBN         = {9783031692314},
  ISSN         = {1611-3349},
  booktitle    = {Financial Cryptography and Data Security. FC 2024 International Workshops},
  publisher    = {Springer Nature Switzerland},
  author       = {Zhang, Haoqian and Yeo, Michelle and Estrada-Gali{\~n}es, Vero and Ford, Bryan},
  year         = {2024},
  month        = Nov,
  pages        = {170--188}
}

@inproceedings{novakovic2026cryptobazaar,
  series       = {NDSS 2026},
  title        = {Cryptobazaar: Private Sealed-bid Auctions at Scale},
  booktitle    = {Proceedings 2026 Network and Distributed System Security Symposium},
  publisher    = {Internet Society},
  author       = {Novakovic, Andrija and Kavousi, Alireza and Gurkan, Kobi and Jovanovic, Philipp},
  year         = {2026},
  collection   = {NDSS 2026}
}

@inbook{huang2025optimizing,
  title        = {Optimizing Liveness for Blockchain-Based Sealed-Bid Auctions in Rational Settings},
  ISBN         = {9783032004956},
  ISSN         = {1611-3349},
  booktitle    = {Financial Cryptography and Data Security. FC 2025 International Workshops},
  publisher    = {Springer Nature Switzerland},
  author       = {Huang, Maozhou and Su, Xiangyu and Larangeira, Mario and Tanaka, Keisuke},
  year         = {2025},
  month        = Nov,
  pages        = {1--28}
}

@misc{alpos2026censorship,
  author       = {Alpos, Orestis and Heimbach, Lioba and Nayak, Kartik and Wadhwa, Sarisht},
  title        = {Censorship-Resistant Sealed-Bid Auctions on Blockchains},
  publisher    = {arXiv},
  year         = {2026},
  copyright    = {arXiv.org perpetual, non-exclusive license}
}

@inproceedings{gebele2025sealed,
  series       = {DeFi '25},
  title        = {Sealed-Bid Auctions via TEE-Backed Confidential Compute Blockchains},
  booktitle    = {Proceedings of the 2025 Workshop on Decentralized Finance and Security},
  publisher    = {ACM},
  author       = {Gebele, Jonas and Mutzel, Timm and {\"O}z, Burak and Matthes, Florian},
  year         = {2025},
  month        = Oct,
  pages        = {16--26},
  collection   = {DeFi '25}
}

@inproceedings{abraham2006distributed,
  author       = {Ittai Abraham and
                  Danny Dolev and
                  Rica Gonen and
                  Joseph Y. Halpern},
  editor       = {Eric Ruppert and
                  Dahlia Malkhi},
  title        = {Distributed computing meets game theory: robust mechanisms for rational
                  secret sharing and multiparty computation},
  booktitle    = {Proceedings of the Twenty-Fifth Annual {ACM} Symposium on Principles
                  of Distributed Computing, {PODC} 2006, Denver, CO, USA, July 23-26,
                  2006},
  pages        = {53--62},
  publisher    = {{ACM}},
  year         = {2006},
  doi          = {10.1145/1146381.1146393}
}

@article{halpern2015algorithm,
  author       = {Joseph Y. Halpern and
                  Rafael Pass},
  title        = {Algorithmic rationality: Game theory with costly computation},
  journal      = {J. Econ. Theory},
  volume       = {156},
  pages        = {246--268},
  year         = {2015},
  doi          = {10.1016/J.JET.2014.04.007}
}

@inproceedings{asharov2011towards,
  author       = {Gilad Asharov and
                  Ran Canetti and
                  Carmit Hazay},
  editor       = {Kenneth G. Paterson},
  title        = {Towards a Game Theoretic View of Secure Computation},
  booktitle    = {Advances in Cryptology - {EUROCRYPT} 2011 - 30th Annual International
                  Conference on the Theory and Applications of Cryptographic Techniques,
                  Tallinn, Estonia, May 15-19, 2011. Proceedings},
  series       = {Lecture Notes in Computer Science},
  volume       = {6632},
  pages        = {426--445},
  publisher    = {Springer},
  year         = {2011},
  doi          = {10.1007/978-3-642-20465-4\_24}
}

@article{wang2012harnessing,
  title={Harnessing the cloud for securely outsourcing large-scale systems of linear equations},
  author={Wang, Cong and Ren, Kui and Wang, Jia and Wang, Qian},
  journal={IEEE Transactions on Parallel and Distributed Systems},
  volume={24},
  number={6},
  pages={1172--1181},
  year={2012},
  publisher={IEEE}
}

@inproceedings{gennaro2010non,
  title={Non-interactive verifiable computing: Outsourcing computation to untrusted workers},
  author={Gennaro, Rosario and Gentry, Craig and Parno, Bryan},
  booktitle={Advances in Cryptology--CRYPTO 2010: 30th Annual Cryptology Conference, Santa Barbara, CA, USA, August 15-19, 2010. Proceedings 30},
  pages={465--482},
  year={2010},
  organization={Springer}
}

@article{shan2018practical,
  title={Practical secure computation outsourcing: A survey},
  author={Shan, Zihao and Ren, Kui and Blanton, Marina and Wang, Cong},
  journal={ACM Computing Surveys (CSUR)},
  volume={51},
  number={2},
  pages={1--40},
  year={2018},
  publisher={ACM New York, NY, USA}
}

@article{wang2024mechanism,
  author       = {Wenhao Wang and
                  Lulu Zhou and
                  Aviv Yaish and
                  Fan Zhang and
                  Ben Fisch and
                  Benjamin Livshits},
  title        = {Mechanism Design for ZK-Rollup Prover Markets},
  journal      = {CoRR},
  volume       = {abs/2404.06495},
  year         = {2024},
  eprinttype    = {arXiv}
}

@article{roughgarden2010algorithmic,
  title={Algorithmic game theory},
  author={Roughgarden, Tim},
  journal={Communications of the ACM},
  volume={53},
  number={7},
  pages={78--86},
  year={2010},
  publisher={ACM New York, NY, USA}
}

@inproceedings{bunz2018bulletproofs,
  author       = {Benedikt B{\"{u}}nz and
                  Jonathan Bootle and
                  Dan Boneh and
                  Andrew Poelstra and
                  Pieter Wuille and
                  Gregory Maxwell},
  title        = {Bulletproofs: Short Proofs for Confidential Transactions and More},
  booktitle    = {2018 {IEEE} Symposium on Security and Privacy, {SP} 2018, Proceedings,
                  21-23 May 2018, San Francisco, California, {USA}},
  pages        = {315--334},
  publisher    = {{IEEE} Computer Society},
  year         = {2018}
}

@inproceedings{thyagarajan2020verifiable,
  author       = {Sri Aravinda Krishnan Thyagarajan and
                  Adithya Bhat and
                  Giulio Malavolta and
                  Nico D{\"{o}}ttling and
                  Aniket Kate and
                  Dominique Schr{\"{o}}der},
  editor       = {Jay Ligatti and
                  Xinming Ou and
                  Jonathan Katz and
                  Giovanni Vigna},
  title        = {Verifiable Timed Signatures Made Practical},
  booktitle    = {{CCS} '20: 2020 {ACM} {SIGSAC} Conference on Computer and Communications
                  Security, Virtual Event, USA, November 9-13, 2020},
  pages        = {1733--1750},
  publisher    = {{ACM}},
  year         = {2020}
}

@inproceedings{wadhwa2022he,
  author       = {Sarisht Wadhwa and
                  Jannis Stoeter and
                  Fan Zhang and
                  Kartik Nayak},
  title        = {He-HTLC: Revisiting Incentives in {HTLC}},
  booktitle    = {30th Annual Network and Distributed System Security Symposium, {NDSS}
                  2023, San Diego, California, USA, February 27 - March 3, 2023},
  publisher    = {The Internet Society},
  year         = {2023}
}

@inproceedings{kosba2016hawk,
  title={Hawk: The blockchain model of cryptography and privacy-preserving smart contracts},
  author={Kosba, Ahmed and Miller, Andrew and Shi, Elaine and Wen, Zikai and Papamanthou, Charalampos},
  booktitle={2016 IEEE symposium on security and privacy (SP)},
  pages={839--858},
  year={2016},
  organization={IEEE}
}

@article{li2021blockchain,
  title={A Blockchain-Based Sealed-Bid e-Auction Scheme with Smart Contract and Zero-Knowledge Proof},
  author={Li, Honglei and Xue, Weilian},
  journal={Security and Communication Networks},
  volume={2021},
  number={1},
  pages={5523394},
  year={2021},
  publisher={Wiley Online Library}
}

@inproceedings{david2022fast,
  title={FAST: fair auctions via secret transactions},
  author={David, Bernardo and Gentile, Lorenzo and Pourpouneh, Mohsen},
  booktitle={International Conference on Applied Cryptography and Network Security},
  pages={727--747},
  year={2022},
  organization={Springer}
}

@article{vickrey1961counterspeculation,
  title={Counterspeculation, auctions, and competitive sealed tenders},
  author={Vickrey, William},
  journal={The Journal of finance},
  volume={16},
  number={1},
  pages={8--37},
  year={1961},
  publisher={JSTOR}
}

@article{abulkasim2021quantum,
title = {Quantum-based privacy-preserving sealed-bid auction on the blockchain},
journal = {Optik},
volume = {242},
pages = {167039},
year = {2021},
issn = {0030-4026},
author = {Hussein Abulkasim and Atefeh Mashatan and Shohini Ghose},
}

@article{xiong2019anonymous,
   title={Anonymous Auction Protocol Based on Time-Released Encryption Atop Consortium BlockChain},
   volume={09},
   ISSN={2231-1548},
   number={01},
   journal={International Journal of Advanced Information Technology},
   publisher={Academy and Industry Research Collaboration Center (AIRCC)},
   author={Xiong, Jie and Wang, Qi},
   year={2019},
   month=feb, pages={01–16} }

@inproceedings{freitag2021nonmalleadle,
  author       = {Cody Freitag and
                  Ilan Komargodski and
                  Rafael Pass and
                  Naomi Sirkin},
  editor       = {Kobbi Nissim and
                  Brent Waters},
  title        = {Non-malleable Time-Lock Puzzles and Applications},
  booktitle    = {Theory of Cryptography - 19th International Conference, {TCC} 2021,
                  Raleigh, NC, USA, November 8-11, 2021, Proceedings, Part {III}},
  series       = {Lecture Notes in Computer Science},
  volume       = {13044},
  pages        = {447--479},
  publisher    = {Springer},
  year         = {2021},
  bibsource    = {dblp computer science bibliography, https://dblp.org}
}

@misc{abadi2023delegated,
  title = {Delegated {{Time-Lock Puzzle}}},
  author = {Abadi, Aydin and Ristea, Dan and Murdoch, Steven J.},
  year = {2023},
  number = {2023/1183},
  urldate = {2024-01-23}
}

@incollection{brakerski2019leveraging,
  title = {Leveraging {{Linear Decryption}}: {{Rate-1 Fully-Homomorphic Encryption}} and {{Time-Lock Puzzles}}},
  shorttitle = {Leveraging {{Linear Decryption}}},
  booktitle = {Theory of {{Cryptography}}},
  author = {Brakerski, Zvika and D{\"o}ttling, Nico and Garg, Sanjam and Malavolta, Giulio},
  editor = {Hofheinz, Dennis and Rosen, Alon},
  year = {2019},
  volume = {11892},
  pages = {407--437},
  publisher = {{Springer International Publishing}},
  address = {{Cham}},
  urldate = {2024-01-23},
  langid = {english}
}

@incollection{malavolta2019homomorphic,
  title = {Homomorphic {{Time-Lock Puzzles}} and {{Applications}}},
  booktitle = {Advances in {{Cryptology}} {\textendash} {{CRYPTO}} 2019},
  author = {Malavolta, Giulio and Thyagarajan, Sri Aravinda Krishnan},
  editor = {Boldyreva, Alexandra and Micciancio, Daniele},
  year = {2019},
  volume = {11692},
  pages = {620--649},
  publisher = {{Springer International Publishing}},
  address = {{Cham}},
  urldate = {2024-01-23},
  langid = {english}
}

@misc{rivest1996timelock,
  title = {Time-Lock Puzzles and Timed-Release {{Crypto}}},
  author = {Rivest, Ronald L and Shamir, Adi and Wagner, David A},
  year = {1996},
  howpublished = {MIT Laboratory for Computer Science, Technical Memo MIT/LCS/TR-684},
  langid = {english}
}

@inproceedings{tyagi2023riggs,
  title={Riggs: Decentralized sealed-bid auctions},
  author={Tyagi, Nirvan and Arun, Arasu and Freitag, Cody and Wahby, Riad and Bonneau, Joseph and Mazi{\`e}res, David},
  booktitle={Proceedings of the 2023 ACM SIGSAC Conference on Computer and Communications Security},
  pages={1227--1241},
  year={2023}
}

@misc{glaeser2023cicada,
  title = {Cicada: {{A}} Framework for Private Non-Interactive on-Chain Auctions and Voting},
  shorttitle = {Cicada},
  author = {Glaeser, Noemi and Seres, Istv{\'a}n Andr{\'a}s and Zhu, Michael and Bonneau, Joseph},
  year = {2023},
  number = {2023/1473},
  urldate = {2024-01-23}
}

@inproceedings{thyagarajan2021opensquare,
  title = {{{OpenSquare}}: {{Decentralized Repeated Modular Squaring Service}}},
  shorttitle = {{{OpenSquare}}},
  booktitle = {Proceedings of the 2021 {{ACM SIGSAC Conference}} on {{Computer}} and {{Communications Security}}},
  author = {Thyagarajan, Sri Aravinda Krishnan and Gong, Tiantian and Bhat, Adithya and Kate, Aniket and Schr{\"o}der, Dominique},
  year = {2021},
  month = nov,
  pages = {3447--3464},
  publisher = {{ACM}},
  address = {{Virtual Event Republic of Korea}},
  urldate = {2024-01-23},
  isbn = {978-1-4503-8454-4},
  langid = {english}
}

@misc{abadi2023decentralised,
  title = {Decentralised {{Repeated Modular Squaring Service Revisited}}: {{Attack}} and {{Mitigation}}},
  shorttitle = {Decentralised {{Repeated Modular Squaring Service Revisited}}},
  author = {Abadi, Aydin and Murdoch, Steven J.},
  year = {2023},
  number = {2023/1347},
  urldate = {2024-01-23}
}

@inproceedings{chawla2012crowdsourcing,
    title = { Optimal Crowdsourcing Contests},
    author = { Chawla, Shuchi and Hartline, Jason D. and Sivan, Balasubramanian},
    editor       = {Yuval Rabani},
  booktitle    = {Proceedings of the Twenty-Third Annual {ACM-SIAM} Symposium on Discrete Algorithms, {SODA} 2012, Kyoto, Japan, January 17-19, 2012},
  pages        = {856--868},
  publisher    = {{SIAM}},
  year         = {2012},
  bibsource    = {dblp computer science bibliography, https://dblp.org}
}

@article{myerson1981optimal,
  title={Optimal auction design},
  author={Myerson, Roger B},
  journal={Mathematics of operations research},
  volume={6},
  number={1},
  pages={58--73},
  year={1981},
  publisher={INFORMS}
}

@article{akbarpour2020credible,
  title={Credible auctions: A trilemma},
  author={Akbarpour, Mohammad and Li, Shengwu},
  journal={Econometrica},
  volume={88},
  number={2},
  pages={425--467},
  year={2020},
  publisher={Wiley Online Library}
}

@inproceedings{blass2018strain,
  title={Strain:A secure auction for blockchains},
  author={Blass, Erik-Oliver and Kerschbaum, Florian},
  booktitle={European Symposium on Research in Computer Security},
  pages={87--110},
  year={2018},
  organization={Springer}
}

@inproceedings{braghin2020designing,
  title={Designing smart-contract based auctions},
  author={Braghin, Chiara and Cimato, Stelvio and Damiani, Ernesto and Baronchelli, Michael},
  booktitle={Security with Intelligent Computing and Big-data Services: Proceedings of the Second International Conference on Security with Intelligent Computing and Big Data Services (SICBS-2018) 2},
  pages={54--64},
  year={2020},
  organization={Springer}
}

@inproceedings{boneh2000timed,
  title={Timed commitments},
  author={Boneh, Dan and Naor, Moni},
  booktitle={Annual international cryptology conference},
  pages={236--254},
  year={2000},
  organization={Springer}
}

@inproceedings{katz2020security,
  title={On the security of time-lock puzzles and timed commitments},
  author={Katz, Jonathan and Loss, Julian and Xu, Jiayu},
  booktitle={Theory of Cryptography: 18th International Conference, TCC 2020, Durham, NC, USA, November 16--19, 2020, Proceedings, Part III 18},
  pages={390--413},
  year={2020},
  organization={Springer}
}

@article{chen2022sbrac,
  title={SBRAC: Blockchain-based sealed-bid auction with bidding price privacy and public verifiability},
  author={Chen, Biwen and Li, Xue and Xiang, Tao and Wang, Peng},
  journal={Journal of Information Security and Applications},
  volume={65},
  pages={103082},
  year={2022},
  publisher={Elsevier}
}

@inproceedings{krol2020pastrami,
  title={PASTRAMI: privacy-preserving, auditable, Scalable \& Trustworthy Auctions for multiple items},
  author={Kr{\'o}l, Michal and Sonnino, Alberto and Tasiopoulos, Argyrios and Psaras, Ioannis and Rivi{\`e}re, Etienne},
  booktitle={Proceedings of the 21st International Middleware Conference},
  pages={296--310},
  year={2020}
}

@inproceedings{lu2021blockchain,
  title={Blockchain-based sealed-bid domain name auction protocol},
  author={Lu, Genhua and Zhang, Yi and Lu, Zhongxiang and Shao, Jun and Wei, Guiyi},
  booktitle={Applied Cryptography in Computer and Communications: First EAI International Conference, AC3 2021, Virtual Event, May 15-16, 2021, Proceedings 1},
  pages={25--38},
  year={2021},
  organization={Springer}
}

@inproceedings{constantinides2021block,
  title={Block auction: A general blockchain protocol for privacy-preserving and verifiable periodic double auctions},
  author={Constantinides, Theodoros and Cartlidge, John},
  booktitle={2021 IEEE International Conference on Blockchain (Blockchain)},
  pages={513--520},
  year={2021},
  organization={IEEE}
}

@inproceedings{desai2019hybrid,
  title={A hybrid blockchain architecture for privacy-enabled and accountable auctions},
  author={Desai, Harsh and Kantarcioglu, Murat and Kagal, Lalana},
  booktitle={2019 IEEE International Conference on Blockchain (Blockchain)},
  pages={34--43},
  year={2019},
  organization={IEEE}
}

@inproceedings{galal2019verifiable,
  title={Verifiable sealed-bid auction on the ethereum blockchain},
  author={Galal, Hisham S and Youssef, Amr M},
  booktitle={Financial Cryptography and Data Security: FC 2018 International Workshops, BITCOIN, VOTING, and WTSC, Nieuwpoort, Curacao, March 2, 2018, Revised Selected Papers 22},
  pages={265--278},
  year={2019},
  organization={Springer}
}

@article{sharma2021anonymous,
  title={Anonymous sealed-bid auction on ethereum},
  author={Sharma, Gaurav and Verstraeten, Denis and Saraswat, Vishal and Dricot, Jean-Michel and Markowitch, Olivier},
  journal={Electronics},
  volume={10},
  number={19},
  pages={2340},
  year={2021},
  publisher={MDPI}
}

@inproceedings{desai2021secauctee,
  title={SECAUCTEE: Securing auction smart contracts using trusted execution environments},
  author={Desai, Harsh and Kantarcioglu, Murat},
  booktitle={2021 IEEE International Conference on Blockchain (Blockchain)},
  pages={448--455},
  year={2021},
  organization={IEEE}
}

@inproceedings{galal2020trustee,
  title={Trustee: full privacy preserving vickrey auction on top of ethereum},
  author={Galal, Hisham S and Youssef, Amr M},
  booktitle={Financial Cryptography and Data Security: FC 2019 International Workshops, VOTING and WTSC, St. Kitts, St. Kitts and Nevis, February 18--22, 2019, Revised Selected Papers 23},
  pages={190--207},
  year={2020},
  organization={Springer}
}

@inproceedings{bangerter2005efficient,
  title={Efficient proofs of knowledge of discrete logarithms and representations in groups with hidden order},
  author={Bangerter, Endre and Camenisch, Jan and Maurer, Ueli},
  booktitle={International Workshop on Public Key Cryptography},
  pages={154--171},
  year={2005},
  organization={Springer}
}

@inproceedings{boneh2019batching,
  title={Batching techniques for accumulators with applications to IOPs and stateless blockchains},
  author={Boneh, Dan and B{\"u}nz, Benedikt and Fisch, Ben},
  booktitle={Advances in Cryptology--CRYPTO 2019: 39th Annual International Cryptology Conference, Santa Barbara, CA, USA, August 18--22, 2019, Proceedings, Part I 39},
  pages={561--586},
  year={2019},
  organization={Springer}
}

@inproceedings{castagnos2020bandwidth,
  title={Bandwidth-efficient threshold EC-DSA},
  author={Castagnos, Guilhem and Catalano, Dario and Laguillaumie, Fabien and Savasta, Federico and Tucker, Ida},
  booktitle={IACR International Conference on Public-Key Cryptography},
  pages={266--296},
  year={2020},
  organization={Springer}
}

@inproceedings{thyagarajan2021efficient,
  title={Efficient CCA timed commitments in class groups},
  author={Thyagarajan, Sri Aravinda Krishnan and Castagnos, Guilhem and Laguillaumie, Fabian and Malavolta, Giulio},
  booktitle={Proceedings of the 2021 ACM SIGSAC Conference on Computer and Communications Security},
  pages={2663--2684},
  year={2021}
}

@inproceedings{chvojka2021versatile,
  title={Versatile and sustainable timed-release encryption and sequential time-lock puzzles},
  author={Chvojka, Peter and Jager, Tibor and Slamanig, Daniel and Striecks, Christoph},
  booktitle={European Symposium on Research in Computer Security},
  pages={64--85},
  year={2021},
  organization={Springer}
}

@inproceedings{choi2023bicorn,
  title={Bicorn: An optimistically efficient distributed randomness beacon},
  author={Choi, Kevin and Arun, Arasu and Tyagi, Nirvan and Bonneau, Joseph},
  booktitle={International Conference on Financial Cryptography and Data Security},
  pages={235--251},
  year={2023},
  organization={Springer}
}

@inproceedings{paillier1999public,
  title={Public-key cryptosystems based on composite degree residuosity classes},
  author={Paillier, Pascal},
  booktitle={International conference on the theory and applications of cryptographic techniques},
  pages={223--238},
  year={1999},
  organization={Springer}
}

@inproceedings{chawla2013auctions,
  title={Auctions with unique equilibria},
  author={Chawla, Shuchi and Hartline, Jason D},
  booktitle={Proceedings of the fourteenth ACM conference on Electronic commerce},
  pages={181--196},
  year={2013}
}
\clearpage
\appendix
\section*{Appendix}

\section{Preliminaries: Timed Commitment}
\label{appendix:prelim}

We introduce the timed commitment scheme, a key component of our protocol, which allows bidders to commit to a bid while keeping it sealed and confidential until the designated reveal time.

\begin{definition}[Timed Commitment~\cite{rivest1996timelock}]
	A timed commitment scheme is a tuple of three algorithms $<PSetup, PGen, PSolve>$, defined as follows:
	\begin{enumerate}
		\item $pp \leftarrow PSetup(1^{\lambda}, \calT)$: An algorithm that takes a security parameter $1^{\lambda}$, and a hardness parameter $\calT$ as input, and outputs public parameters $pp$.
		\item $Z \leftarrow PGen(pp, e)$: A probabilistic algorithm that takes as input a public parameter $pp$ and a secret $e \in \calE$, where $\calE$ is the secret space, and outputs a commitment $Z$.
		\item $e \leftarrow PSolve(pp, Z)$: A deterministic algorithm that takes as input a commitment $Z$ and outputs a decryption secret $e$.
	\end{enumerate}
\end{definition}

A timed commitment scheme must satisfy the following two properties:

\smallskip
\noindent\textbf{Correctness.}
    For all $\lambda \in \NN$ and all polynomials $\calT$ in $\lambda$, let $pp = PSetup(1^{\lambda}, \calT)$.
    A time-lock puzzle scheme is correct if, for all secret $e \in \calE$, it always holds that $e = PSolve(PGen(pp, e))$.

\smallskip
\noindent\textbf{Security.}
	A timed commitment scheme is secure with gap $\epsilon$ if there exists a polynomial $\tilde{\calT}(\cdot)$ such that for all polynomials $\calT(\cdot) \geq \tilde{\calT}(\cdot)$ and every polynomial-size adversary $\calA = \{\calA_{\lambda}\}_{\lambda \in \NN}$ of depth $\leq \calT^{\epsilon}(\lambda)$, there exists a negligible function $\mu(\cdot)$, such that for all $\lambda \in \NN$, all secret pairs $e^0, e^1 \in \calE$, and a random bit $b \sample \{0,1\}$, it holds that
	\begin{align*}
		\left|Pr[b' = b : e \leftarrow e^{b},\ Z \leftarrow PGen(PSetup(1^{\lambda}, \calT(\lambda)), e),\ b' \leftarrow \calA(Z)] - \tfrac{1}{2}\right| \leq \mu(\lambda).
	\end{align*}

Additionally, we also consider other important functionalities of time-lock puzzles, including homomorphism, aggregatability, and non-malleability.

\smallskip
\noindent\textbf{Homomorphic.} The homomorphic property of timed commitment was first introduced by Malavolta et al.~\cite{malavolta2019homomorphic}. It allows specific types of homomorphic operations on encrypted commitments, ensuring both correctness and security. The definition is as follows:

\noindent\textbf{$Z' \leftarrow Eval(pp, C, Z_1, ..., Z_k)$.} An algorithm takes public parameters, a circuit $C: \calE^k \rightarrow \calE$, and commitments $Z_1, ..., Z_k$ as input, and outputs a commitment $Z'$ such that the corresponding secret $e' = C(e_1, ..., e_k)$.

\smallskip
\noindent\textbf{Aggregatability.} We use the following aggregatability notion: a manager can aggregate multiple puzzles into one while maintaining correctness and security. It should contain two algorithms $PAggr$ and $PASolve$:

\noindent\textbf{$Z^{A} \leftarrow PAggr(pp, M, Z_1, ..., Z_m)$.} An algorithm takes public parameters and commitments $Z_1, \dots, Z_m$ as input, and outputs one commitment $Z^{A}$.

\noindent\textbf{$(e_1, ..., e_m) \leftarrow PASolve(pp, M, Z^{A}, m)$.} An algorithm takes public parameters, the capacity $M$, aggregate commitment $Z^{A}$ and the number of aggregated puzzles $m$ as input, and outputs $m$ secrets satisfying $e_i = PSolve(pp, Z_i)$.

\smallskip
\noindent\textbf{Non-malleability.} The non-malleability of timed commitment means that even if an adversary intercepts a commitment, they cannot create a new commitment whose solution would reveal any information about the original solution. This property is crucial for maintaining the security and integrity of timed commitment, particularly in scenarios where they are used to delay access to information or to secure timed releases~\cite{freitag2021nonmalleadle}.

\section{Security Analysis of Aggregatable Time-Lock Puzzles} \label{appendix:ATLP}

\begin{definition}[Aggregatable time-lock puzzle]
\label{def:atlp-syntax}
An aggregatable homomorphic time-lock puzzle for secret space $\mathcal{E}$ consists of four algorithms:
\begin{itemize}
    \item $pp\leftarrow PSetup(1^\lambda,\mathcal{T})$ outputs public parameters for security parameter $\lambda$ and hardness parameter $\mathcal{T}$;
    \item $Z_i\leftarrow PGen(pp,M,i,e_i)$ outputs the puzzle of bidder $i\in[M]$ of the current auction (with capacity $M$) for secret $e_i\in\mathcal{E}$;
    \item $Z\leftarrow PAggr(pp,M,Z_1,\dots,Z_m)$ ($m\le M$) outputs an aggregate puzzle;
    \item $(e_1,\dots,e_m)\leftarrow PASolve(pp,M,Z,m)$ outputs all embedded secrets.
\end{itemize}
\end{definition}

\begin{definition}[Aggregatable correctness]
\label{def:atlp-correctness}
For every polynomial $m$ and all secrets $e_1,\dots,e_m\in\mathcal{E}$,
\[
    \Pr\left[
        PASolve(pp,M,PAggr(pp,M,Z_1,\dots,Z_m),m)=(e_1,\dots,e_m)
    \right]=1,
\]
where the probability is over $PSetup$ and the randomness of $PGen$.
\end{definition}

\begin{definition}[Timed aggregate security]
\label{def:atlp-security}
Consider an adversary $\mathcal{A}=(\mathcal{A}_1,\mathcal{A}_2)$ of depth at most $\mathcal{T}^{\epsilon}(\lambda)$. In experiment $b\in\{0,1\}$, $\mathcal{A}_1$ outputs a target coordinate $k\in[m]$, secrets $e_k^0,e_k^1$, auxiliary secrets for the other coordinates, and state $\tau$. The challenger creates puzzles for all secrets, gives both the individual puzzles and their aggregate to $\mathcal{A}_2(\tau)$, and receives a bit $b'$. The scheme is secure with gap $\epsilon$ if, for every such $\mathcal{A}$, $k$, and secret pair,
\[
    \left|\Pr[b'=1\mid b=0]-\Pr[b'=1\mid b=1]\right|\le \mu(\lambda)
\]
for a negligible function $\mu$.
\end{definition}

We first verify coordinate-wise correctness. For any $k\in[m]$, let the result computed by $PASolve$ be $\tilde{e}_k$. We show that $\tilde{e}_k=e_k$.

Let $R=\sum_{i=1}^m r_i$. All identities below hold in $\mathbb{Z}_{N^{M+1}}^*$, with $h$ reduced modulo $N^{M+1}$: since each $u_i = g^{r_i}$, we have $U = g^{R}$, hence $X = U^{2^\mathcal{T}} = h^{R}$ exactly, and
\[
	W = h^{NR}\prod_{i=1}^m (1+N^i)^{e_i}.
\]
Therefore $W / X^{N} = \prod_{i=1}^m (1+N^i)^{e_i} \in 1 + N\mathbb{Z}_{N^{M+1}}$, and reducing modulo $N^{k+1}$ (where $(1+N^i)^{e_i} \equiv 1$ for $i > k$) gives
\begin{align*}
	\tilde{e}_k &= \frac{\left( \frac{W}{(X^N) \cdot \prod_{i=1}^{k-1}(1+N^i)^{e_i}} \right) \bmod N^{k+1} - 1}{N^k} \\
	&= \frac{(1+N^k)^{e_k} \bmod N^{k+1} - 1}{N^k} = \frac{e_k N^k}{N^k} = e_k,
\end{align*}
where the last step uses $(1+N^k)^{e_k} \equiv 1 + e_k N^k \pmod{N^{k+1}}$, and the recovered digits are unique by Lemma~\ref{lemma:group}.

The security of the aggregatable time-lock puzzle scheme relies on the following foundational assumptions.
The sequential-squaring assumption originates from Rivest et al.~\cite{rivest1996timelock}, DDH from Diffie and Hellman~\cite{diffie1976new}, and DCR from Paillier~\cite{paillier1999public}. Malavolta et al.~\cite{malavolta2019homomorphic} use and analyze these assumptions in the homomorphic time-lock-puzzle setting.

\begin{assumption}[Sequential Squaring~\cite{rivest1996timelock}]
    Sequential Squaring: Let $N$ be a uniform strong RSA integer, $g$ be a generator of $\mathbb{J}_N$, and $T(\cdot)$ be a polynomial. Then there exists some $\epsilon$ such that for every polynomial-size adversary $A = \{A_{\lambda}\}_{\lambda \in \NN}$ whose depth is bounded from above by $T^{\epsilon}(\lambda)$, there exists a negligible function $\mu(\cdot)$ such that
    \[\Pr\left[
    \begin{array}{l}
    b' = b \\
    \end{array}
    %\middle|
    :
    \begin{array}{l}
    x \overset{\$}{\leftarrow} \mathbb{J}_N; 
    b \overset{\$}{\leftarrow} \{0,1\} \\
    \text{if } b = 0 \text{ then } y \overset{\$}{\leftarrow} \mathbb{J}_{N} \\
    \text{if } b = 1 \text{ then } y := x^{2^{T(\lambda)}} \\
    b' \leftarrow \mathcal{A}(N,g, T(\lambda),x, y)
    \end{array}
    \right] \leq \frac{1}{2} + \mu(\lambda).\]
\end{assumption}

\begin{assumption}[Decisional Composite Residuosity~\cite{paillier1999public}] %2-k
    Decisional Composite Residuosity on $\mathbb{Z}^*_{N^{k}}$: Let $N$ be a uniform strong RSA integer. Then for every polynomial-size adversary $A = \{A_{\lambda}\}_{\lambda \in \NN}$ there exists a negligible function $\mu(\cdot)$ such that
    \[\Pr\left[
    \begin{array}{l}
    b' = b \\
    \end{array}
    %\middle|
    :
    \begin{array}{l}
    x \overset{\$}{\leftarrow} \mathbb{Z}^*_N;
    b \overset{\$}{\leftarrow} \{0,1\} \\
    \text{if } b = 0 \text{ then } y \overset{\$}{\leftarrow} \mathbb{Z}^*_{N^k} \\
    \text{if } b = 1 \text{ then } y := x^N (\text{mod } N^k) \\
    b' \leftarrow \mathcal{A}(N, y)
    \end{array}
    \right] \leq \frac{1}{2} + \mu(\lambda).\]

\end{assumption}

\begin{assumption}[Decisional Diffie-Hellman~\cite{diffie1976new}]
    Decisional Diffie-Hellman: Let $N$ be a uniform strong RSA integer and $g$ be a generator of $\mathbb{J}_N$. Then for every polynomial-size adversary $A = \{A_{\lambda}\}_{\lambda \in \NN}$ there exists a negligible function $\mu(\cdot)$ such that
    \[\Pr\left[
    \begin{array}{l}
    b' = b \\
    \end{array}
    %\middle|
    :
    \begin{array}{l}
    (x,y) \overset{\$}{\leftarrow} \{1, ..., \varphi(N)/2 \};
    b \overset{\$}{\leftarrow} \{0,1\} \\
    \text{if } b = 0 \text{ then } z \overset{\$}{\leftarrow} \{1, ..., \varphi(N)/2 \} \\
    \text{if } b = 1 \text{ then } z := x\cdot y (\text{mod } \varphi(N)/2) \\
    b' \leftarrow \mathcal{A}(N, g, g^x, g^y, g^z)
    \end{array}
    \right] \leq \frac{1}{2} + \mu(\lambda).\]
\end{assumption}

By introducing the Decisional Composite Residuosity assumption on $\mathbb{Z}^*_{N^k}$ for $k \in \{2, \dots, M+1\}$, we can reduce the security of our scheme to the scheme studied in~\cite{malavolta2019homomorphic}, and then utilize their Theorem 1 to establish the security proof.

\begin{theorem}[Aggregatable correctness]
\label{thm:atlp-correctness}
For every $m$ and all $e_1,\dots,e_m\in\{0,\dots,N-1\}$, the scheme in Fig.~\ref{fig:AHTLP} satisfies aggregatable correctness.
\end{theorem}

\begin{theorem}[Timed aggregate security]
\label{thm:atlp-security}
	Let $N$ be a strong RSA integer. If the sequential squaring assumption and DDH assumption hold over $\mathbb{J}_N$, and the DCR assumption holds over $\mathbb{Z}^*_{N^k}$ for $k \in \{2, \dots, M+1\}$, then the scheme in Fig.~\ref{fig:AHTLP} satisfies timed aggregate security.
\end{theorem}

\begin{proof}
Theorem~\ref{thm:atlp-correctness} is proved by the coordinate-extraction calculation.
All identities hold in $\mathbb{Z}_{N^{M+1}}^*$, with $h$ reduced modulo $N^{M+1}$.
Let $R=\sum_{i=1}^m r_i$. Then $U=g^R$, $X=U^{2^\mathcal{T}}=h^R$, and
\[
    W=h^{RN}\prod_{i=1}^m(1+N^i)^{e_i}.
\]
Dividing by $X^N$ and the previously recovered factors $(1+N^i)^{e_i}$ leaves $(1+N^k)^{e_k}$ modulo $N^{k+1}$, and the extracted digit is unique by Lemma~\ref{lemma:group}, so the extracted value is $e_k$.

For Theorem~\ref{thm:atlp-security}, we use the single-puzzle timed-privacy statement for our variant (the subgroup $\mathbb{J}_N$, the modulus $N^{M+1}$, and the packed plaintext $(1+N^k)^{e_k}$), which follows from the same hybrid sequence as the proof of Theorem 1 of Malavolta et al.~\cite{malavolta2019homomorphic}: first replace $h = g^{2^{\mathcal{T}}}$ by a uniform element of $\mathbb{J}_N$ via the sequential squaring assumption, then randomize $u = g^{r}$ via DDH over $\mathbb{J}_N$, and finally mask the packed plaintext via DCR at each level up to $M+1$. Under the three assumptions above (with DCR on $\mathbb{Z}^*_{N^\ell}$ for every level $\ell \le M+1$), for every fixed coordinate $k$, an honestly generated puzzle $Z_k$ hides $e_k$ from adversaries of depth at most $\mathcal{T}^{\epsilon}(\lambda)$.

Suppose an adversary distinguishes the aggregate experiments for coordinate $k$. Construct an adversary against the corresponding single-puzzle experiment by forwarding $(e_k^0,e_k^1)$ to its challenger. Upon receiving challenge $Z_k^b$, sample all other puzzles using the auxiliary secrets, aggregate the other puzzles component-wise, $C \leftarrow PAggr(pp, M, \{Z_i\}_{i\neq k})$, and give the adversary both the individual puzzles and the aggregate $Z$ whose two components are the component-wise products of the corresponding components of $Z_k^b$ and $C$. Component-wise multiplication by the fixed $C$ is a bijection, and $C$ has the same distribution in both experiments. Hence the two distinguishers have the same advantage, contradicting the imported single-puzzle privacy statement.
\end{proof}

These assumptions use standard hidden-order-group ingredients. Sequential squaring is the standard time-lock assumption, DDH over the subgroup $\mathbb{J}_N$ of Jacobi-symbol-\(+1\) elements is a standard subgroup assumption, and DCR at higher powers of $N$ is the usual higher-residuosity extension employed by Paillier-type and packed homomorphic constructions. We do not claim that one fixed DCR instance implies security at every higher level: the theorem explicitly assumes DCR over $\mathbb{Z}^*_{N^k}$ for every level \(k\in\{2,\dots,M+1\}\). Increasing \(M\) increases the modulus size and the number of hybrid levels, so concrete parameters should be selected for the largest system capacity \(L\).

\section{Range Proofs, Verification, and Honest Decryption}
\label{appendix:comparison-verification}

\subsection{Range Proofs and Comparison with Other Solutions}
\label{appendix:comparison}

\noindent\textbf{Range proofs for aggregatable time-lock puzzles.}
Aggregatable time-lock puzzles are also compatible with the range proof method presented in~\cite{tyagi2023riggs}. Specifically, the prover must first demonstrate that the secret in the aggregatable time-lock puzzle is equivalent to a Pedersen commitment through a group homomorphism proof as outlined in~\cite{thyagarajan2021efficient}. Subsequently, an efficient range proof for Pedersen commitments can be constructed, as detailed in~\cite{bunz2018bulletproofs}.

\smallskip
\noindent\textbf{Comparing with other possible solutions.}
Our aggregatable time-lock puzzle scheme differs from the packing protocol in Cicada~\cite{glaeser2023cicada}, which does not consider the change in the size of the group $N$ relative to the total number of bidders $m$.
Consider a scenario where the maximum possible value for each bid is $B$.
Cicada's scheme necessitates prior knowledge of~$m$, and requires the selection of a group size $\tilde{N} \geq B^m$ to prevent bid overflow.
Consequently, each commitment in Cicada's scheme is constructed over the group $\mathbb{Z}_{\tilde{N}^2}^*$.
In contrast, our scheme only requires choosing once at setup a group size $N > B$, so that every bid in $[0,B]$ admits a unique coordinate decomposition (Lemma~\ref{lemma:group}); each auction additionally announces a capacity $M$, and commitments in that auction are constructed over the group $\mathbb{Z}_{N^{M+1}}^*$.
When the announced capacity matches the number of bidders ($M = m$), our commitments live in a group of size about $B^{m+1}$, whereas Cicada's commitments for $m$ bidders live in $\mathbb{Z}_{\tilde{N}^2}^*$ with $\tilde{N} \geq B^m$, i.e., a group larger than $B^{2m}$; moreover, Cicada must fix $m$ before choosing $\tilde{N}$, while our setup is independent of the bidder count. This significant reduction in the required group size improves the computational efficiency of our scheme.

\subsection{Outsourcing Versus Alternative Opening Designs}
\label{appendix:opening-alternatives}
Opening $m$ concatenated puzzles separately requires $m$ repeated-squaring chains, or $m\mathcal{T}$ modular squarings in total. The aggregate requires one chain, $\mathcal{T}$ squarings, followed by coordinate extraction costing $O(m\log N)$ group operations. Thus aggregation is not merely a change of data layout: it removes the factor of $m$ from the sequential-opening work.

% \begin{table}[H]
% \centering
% \caption{Opening-market alternatives.}
% \label{tab:opening-alternatives}
% \begin{tabular}{p{3.1cm}p{3.0cm}p{3.0cm}p{3.4cm}}
% \hline
% Design & Opening work & Main incentive issue & Consequence \\
% \hline
% Auctioneer rents one machine & $\mathcal{T}$ & Fixed procurement; private solver costs are not elicited & Under-provision delays revenue; over-provision pays for idle speed \\
% \hline
% Self-open with penalties & $m$ independent openings & Relies on bidders and censorship-resistant inclusion & Penalties may reduce bidder IR; a censored bidder cannot open \\
% \hline
% Concatenate commitments & $m\mathcal{T}$ & Must price and verify each opening separately & Sequential work and transaction cost grow linearly in $m$ \\
% \hline
% Aggregated competition & $\mathcal{T}+O(m\log N)$ & Requires one competitive solver & Public verification selects the first valid opening \\
% \hline
% \end{tabular}
% \end{table}

Renting generic compute is not excluded by our market: a rented machine can participate as a solver. The difference is procurement. A fixed rental commits the auctioneer to one cost-speed choice before knowing solver availability, whereas the competition asks multiple providers to reveal their speed through submission time and pays only the first publicly valid solution.

Encrypted mempools address a complementary propagation-time problem: they can hide transactions while they are propagated for inclusion, but they typically rely on threshold decryption or trusted relay infrastructure and do not by themselves price the cost of forced opening or create an incentive-compatible outsourcing market. The construction here instead provides a contract-verifiable opening path that does not require every bidder to act again after bidding or require a separate threshold committee to release the auction data. Thus it is useful when encrypted-mempool infrastructure is unavailable or when the auction explicitly wants a publicly verifiable fallback opening path.

\subsection{Bidder Verification}
\label{appendix:verification}
During the protocol's execution, a solver may attempt to provide a false result to fraudulently obtain payment.
The primary check is public and automatic: the smart contract verifies the succinct repeated-squaring proof and every coordinate-extraction equation before accepting or rejecting the solution.
Thus a false result cannot produce a payment even if no bidder is watching.

Bidders and other relayers can still submit an invalid solution to the contract before regular execution; the first valid fraud report receives a publicly announced reward $\rho$ from the solver's deposit $D^s$.
Since validity is deterministic and publicly checkable, there is no ``false accusation'' proof to fabricate: an invalid report simply fails the same equations.
Consequently, a bidder need not reveal her own bid to challenge a malformed aggregate solution.

\subsection{Collusion and Bribery Analysis}
\label{appendix:collusion}
\begin{proposition}[Conditional collusion resilience]
\label{prop:collusion}
Suppose at least one participating solver is outside a bidder-solver coalition and follows the protocol. Then the coalition cannot alter or suppress the bid of a non-coalition bidder.
\end{proposition}

\begin{proof}
First, a coalition solver cannot alter a non-coalition bid: any accepted solution must pass the public extraction checks of Theorem~\ref{thm:solver-public}. Second, bidder-specific suppression is not available before opening. By Theorem~\ref{thm:atlp-security}, bids remain hidden, so a coalition cannot identify which non-coalition bidder carries the competing bid. Since all accepted puzzles are multiplied into one aggregate, suppressing that bid requires suppressing the aggregate itself.

If at least one non-coalition solver participates, it can submit the valid aggregate solution and receive the winner-takes-all reward. A bribe that asks this solver to withhold must therefore exceed its expected reward from submitting. The automatic contract execution prevents the auctioneer from quietly ignoring a valid submission. Thus the non-coalition bid is revealed.
\end{proof}

The proposition does not claim protection against a complete solver strike: if every solver colludes and withholds, no computation-procurement mechanism can force the delayed computation to happen. The economic analysis therefore assumes at least one competitive solver. The fraud reward $\rho$ deters false submissions, but no on-chain reward by itself can prove and punish a privately promised bribe to withhold.

\subsection{Honest Decryption of Solvers}
\label{appendix:honest-solver}
Notice that one of the solver strategies is to try to break the cryptographic security of the timed commitment.
However, due to the security guarantees of our aggregatable time-lock puzzle scheme, the probability of a polynomial-time solver successfully doing so is negligible.
As the maximum revenue a solver can gain from an auction is bounded, the expected revenue from attempting to break the scheme is strictly less than the expected revenue from honestly decrypting the commitment.
Therefore, in the subsequent analysis, we only consider the solver to decrypt the timed commitment honestly.

\subsection{Aggregate and Solver-Output Verification}
\label{appendix:public-verification}
Let $\mathcal{R}_i$ denote the bid-well-formedness relation
\[
    \{(pp, M,i,Z,b,r): Z=PGen(pp,M,i,b;r),\ b\in[0,B]\}.
\]
The smart contract accepts a bid only when its NIZK proof for this relation verifies. After bid collection, the contract itself computes
\[
    (U,W)=\left(\prod_{i=1}^m u_i \bmod N^{M+1},\ \prod_{i=1}^m w_i \bmod N^{M+1}\right),
\]
so an auctioneer cannot replace the aggregate with a malformed puzzle.

A solver submits $(X,\vec{b},\pi^{\mathrm{sol}})$. The contract first verifies the succinct proof that $X=U^{2^\mathcal{T}}\pmod{N^{M+1}}$. It then checks, for every $k\in[m]$,
\[
    W \stackrel{?}{=} X^N\cdot
    \left(\prod_{i=1}^k (1+N^i)^{b_i}\right)
    \pmod{N^{k+1}},
\]
where $X^N$ is computed after reducing $X$ modulo $N^{k+1}$. Equivalently, the contract verifies each extraction equation from Fig.~\ref{fig:AHTLP}. A solution is accepted only if all checks pass.

\begin{theorem}[Public aggregate correctness]
\label{thm:aggregate-public}
If all accepted bids have sound well-formedness proofs and the smart contract computes $(U,W)$ from the accepted puzzles, then, except with negligible probability, $Z=(U,W)$ is the correctly formed aggregate of the accepted bids.
\end{theorem}

\begin{proof}
The NIZK soundness implies that every accepted tuple satisfies the stated $PGen$ relation. The contract performs the deterministic products defining $PAggr$ directly on those accepted tuples. Therefore, the resulting tuple should be exactly $PAggr(pp,M,Z_1,\dots,Z_m)$, except in the negligible event that a false well-formedness proof is accepted.
\end{proof}

\begin{theorem}[Public solver-output correctness]
\label{thm:solver-public}
If the succinct repeated-squaring proof is sound and $(U,W)$ is a correctly formed aggregate, then any solution accepted by the contract equals $PASolve(pp,M,(U,W),m)$.
\end{theorem}

\begin{proof}
The proof of $X=U^{2^\mathcal{T}}$ ensures that the solver uses the required repeated-squaring value. The extraction equations checked by the contract are exactly the defining equations of $PASolve$. Thus an accepted vector differs from the vector output by $PASolve$ only in the negligible event that a false succinct proof is accepted.
\end{proof}

% \begin{proof}[Proof of Theorem~\ref{thm:protocol-guarantees}]
% \label{proof:protocol-guarantees}
% Parts (1)--(3) follow from Theorems~\ref{thm:aggregate-public} and~\ref{thm:solver-public}: the two proof systems rule out malformed bids, malformed aggregates, and accepted solver outputs that do not satisfy the extraction equations. For part (4), before a bidder self-opens, its secret is embedded only in its puzzle and in the aggregate. Definition~\ref{def:atlp-security} states that a target coordinate remains hidden from shallow adversaries even when the adversary sees all individual puzzles and the aggregate. Thus neither the auctioneer, other bidders, nor a solver can distinguish the bid before the opening event, except with negligible advantage. The final statement is immediate from the contract rule: the first valid aggregate solution determines all bids; otherwise, all valid self-openings do. If neither event occurs by $t_d$, liveness is not guaranteed.
% \end{proof}

\section{Payment-Proportional and Rank-Based Payment Rules}
\label{appendix:payment-defs}

\begin{definition}[Payment-proportional].
    For a total payment function $V$, a payment rule $\calP^s$ is \emph{payment-proportional} if the payment function satisfies:
    $$
    p_j^s(\vec b, \vec c) = V(\vec b) \cdot x_j(\vec c),
    $$
    where $(x_1(\vec c), \cdots, x_n(\vec c))$ is the payment allocation function, with $\sum_{j=1}^n x_j(\vec c) = 1$.
\end{definition}

\begin{definition}[Rank-based]
    A payment rule $\calP^s$ is rank-based if the function is only determined by $\vec b$ and the rank of $\vec c$ (ranked in decreasing order, equivalently by increasing submission time). Specifically, it requires $\forall i, j \in [n]$,
    $$
    p_j^s(\vec b, \vec c) = p_i^s(\vec b, \vec c'),
    $$
    if the solver $j$'s rank in $\vec{c}$ is the same as solver $i$'s rank in $\vec c'$.
\end{definition}

%It can be verified that our mechanism is payment-proportional rank-based, with a reserve cost.
%To better describe the payment rule that satisfies these properties, we use a definition of the rank-based contest similar to that in~\cite{chawla2012crowdsourcing}.
Essentially, a rank-based payment rule predetermines the division of the payment into $y_1, \dots, y_n$, with $\sum_{j=1}^n y_j = 1$, where $y_j$ is given to the solver whose final speed rank is $j$-th. Rank-based is a natural property in competition design, as in reality, the auctioneer only knows the solver's rank rather than the exact solver's cost. The bidder's bid and the total reward function $V(\vec{v})$ determine the total payment to solvers. Solvers are ordered by their submission time and awarded the corresponding fraction of the payment; specifically, the $j$-th earliest gets a $y_j$ fraction of the total reward $V(\vec{v})$ if submitted before the due time.

Essentially, the property of payment-proportionality separates the influence of the bidder and solver in the mechanism.
Although it is natural to design such payment rules, it is also worth studying whether the auctioneer can achieve a strictly larger revenue with a non-payment-proportional payment rule. We leave this as an open question for future study.

\section{Proofs of Lemmas and Theorems}

\subsection{Proof of Lemma~\ref{lemma:group}}
\label{proof:lemma-group}
\begin{proof}
By the Chinese remainder theorem, $\mathbb{Z}_{N^{m+1}}^* \cong \mathbb{Z}_{p^{m+1}}^* \times \mathbb{Z}_{q^{m+1}}^*$, and the subgroup $1 + N\mathbb{Z}_{N^{m+1}}$ corresponds to $1 + p\mathbb{Z}_{p^{m+1}} \times 1 + q\mathbb{Z}_{q^{m+1}}$. For an odd prime $p$, the group $1 + p\mathbb{Z}_{p^{m+1}}$ is cyclic of order $p^{m}$, generated by $1+p$ (indeed $(1+p)^{p^{j}} \equiv 1 + p^{j+1} \pmod{p^{j+2}}$ with a unit coefficient). Hence $1 + N\mathbb{Z}_{N^{m+1}} \cong \mathbb{Z}_{p^{m}} \times \mathbb{Z}_{q^{m}} \cong \mathbb{Z}_{N^{m}}$ is cyclic of order $N^{m}$, and it is generated by $1+N$ because $(1+N)^{N^{j}} \equiv 1 + N^{j+1} \pmod{N^{j+2}}$ with a unit coefficient.

For the uniqueness of the decomposition, note first that $(1+N^{k})^{t} \equiv 1 + tN^{k} \pmod{N^{k+1}}$ for every $k \ge 1$ (all remaining binomial terms are divisible by $N^{2k} \ge N^{k+1}$), and that $(1+N^{i})^{t} \equiv 1 \pmod{N^{k+1}}$ whenever $i > k$. Suppose $\prod_{k=1}^{m}(1+N^{k})^{t_k} = \prod_{k=1}^{m}(1+N^{k})^{t'_k}$ with all $t_k, t'_k \in \{0,\dots,N-1\}$. Reducing modulo $N^{2}$ gives $1 + t_1 N \equiv 1 + t'_1 N$, hence $t_1 = t'_1$; dividing both sides by $(1+N)^{t_1}$ and reducing modulo $N^{3}$ then yields $t_2 = t'_2$. Iterating level by level gives $t_k = t'_k$ for every $k \in [m]$. Since both the domain $\{0,\dots,N-1\}^{m}$ and the group $1 + N\mathbb{Z}_{N^{m+1}}$ have exactly $N^{m}$ elements, uniqueness implies that the decomposition exists for every $y \in 1 + N\mathbb{Z}_{N^{m+1}}$. The same level-by-level argument also shows how to recover the digits: $t_k$ is the coefficient of $N^{k}$ in $y \cdot \prod_{i<k}(1+N^{i})^{-t_i} \bmod N^{k+1}$, which is exactly the extraction performed by $PASolve$.
\end{proof}

\subsection{Proof of Theorem~\ref{thm:bne}}
\label{proof:thm-bne}
\begin{lemma}[Contest characterization]
\label{lem:contest-characterization}
Consider a symmetric rank-based all-pay contest whose type $l$ is drawn i.i.d.\ from $F$, interim allocation rule $x(l)$ is non-decreasing, prize value is $A$, and effort $a$ yields utility $A x(l,a)-a$. Then the symmetric monotone equilibrium effort is
\[
    a(l)=A\left(lx(l)-\int_1^l x(z)\,dz\right).
\]
The identity is the all-pay payment identity of Myerson~\cite{myerson1981optimal}; for our setting it is verified by the best-response calculation in the proof of Theorem~\ref{thm:bne} below. That awarding the entire prize to the highest type with non-negative virtual value maximizes expected effort among rank-based allocation rules is established by the virtual-surplus computation in the proof of Theorem~\ref{thm:hbw} below; existence of the monotone symmetric equilibrium follows from Lemmas 3 and 4 of Chawla et al.~\cite{chawla2012crowdsourcing}.
\end{lemma}

\textbf{Reduction to our solver market.}
Set the contest type to the solver discount ability $l$, the prize to $A=\mathbb{E}_{\vec v}[V(\vec v)]$, and the contest effort to $a_j=c_j\calT/l_j$. Since the solver utility is
\[
    \mathbb{E}_{\vec v}[V(\vec v)]x_j(c_j)-c_j\calT/l_j,
\]
it has exactly the utility form in Lemma~\ref{lem:contest-characterization}. Substituting $a_j=c_j\calT/l_j$ in the equilibrium identity gives
\[
    c_j(l_j)=\frac{\mathbb{E}_{\vec v}[V(\vec v)]}{\calT}
    \left(l_jx(l_j)-\int_1^{l_j}x(z)\,dz\right).
\]
The second part of the lemma is precisely the winner-takes-all comparison used in the proof of Theorem~\ref{thm:hbw}.
\begin{proof}

% [Theorem 2.7 in \cite{chawla2012crowdsourcing}]: All-pay auction with a reserve price and rank-based (alloc rule) has one unique BNE, which is additionally symmetric.
% Thus this mechanism(high-win-all rule) has BNE.

For any total payment function $V(\vec v)$ and interim payment allocation function $x(l)$ that is non-decreasing on $l$, Lemma~\ref{lem:contest-characterization} gives the symmetric BNE strategy:
$$
c(l) = \frac{\mathbb{E}_{\vec v}[V(\vec v)]}{\calT}(lx(l)-\int_{1}^l x(z)dz).
$$ 
To establish the existence of this symmetric BNE, we demonstrate that for solver 1, playing strategy $c_1 = c(l_1)$ constitutes the best response. Without loss of generality, we focus on the case where $l_1 > \psi_n^{-1}(0)$.
The expected utility of solver 1 is given by
$$
u^s_1(c_1) = \mathbb{E}_{\vec v}[V(\vec v)] \cdot \mathbb{P}[c_1 \ge c(max_{j \ne 1} l_j)] - \frac{c_1 \calT}{l_1}.
$$
We now prove that $c_1 = c(l_1)$ maximizes solver 1's expected utility. 
We first argue that only costs of the form $c(t)$ with $t \ge \psi_n^{-1}(0)$, together with $c_1 = 0$, need to be considered. Recall that a solution is eligible for payment only if it is submitted before the due time $t_d = \calT/\kappa(c(\psi_n^{-1}(0)))$ (Fig.~\ref{fig:sol}). Since $c(l) = 0$ for $l < \psi_n^{-1}(0)$ and $c$ is strictly increasing on $\{l \ge \psi_n^{-1}(0)\}$, a deviation to any cost $c_1 \in (0, c(\psi_n^{-1}(0)))$ submits strictly after $t_d$, is ineligible for payment, and yields utility $-c_1 \calT / l_1 < 0$; it is therefore dominated by $c_1 = 0$. Hence it suffices to consider costs $c_1 = c(t)$ with $t \ge \psi_n^{-1}(0)$, for which $\mathbb{P}[c_1 \ge c(\max_{j \ne 1} l_j)] = F^{n-1}(t)$ (ties occur with probability zero).
Note that $c(l)$ is continuous and strictly increasing in $l$ on $\{l \ge \psi_n^{-1}(0)\}$. Hence, there exists some $t$ such that $c_1 = c(t)$, and it suffices to show that $t = l_1$ maximizes solver 1's utility.

\begin{align}
\frac{l_1}{\calT} \cdot u^s_1(c_1) &= \frac{\mathbb{E}_{\vec v}[V(\vec v)]}{\calT} \cdot \mathbb{P}[c_1\ge c(max_{j \ne 1}l_j)] * l_1 - c_1\\
&= \frac{\mathbb{E}_{\vec v}[V(\vec v)]}{\calT} \cdot \mathbb{P}[t \ge max_{j \ne 1} l_j] \cdot l_1 - c(t) \\
&= \frac{\mathbb{E}_{\vec v}[V(\vec v)]}{\calT}(F^{n-1}(t) \cdot l_1-t x(t) + \int_{1}^t x(z)dz) \\
&= \frac{\mathbb{E}_{\vec v}[V(\vec v)]}{\calT}(F^{n-1}(t) \cdot (l_1-t) + \int_{\psi_n^{-1}(0)}^t F^{n-1}(z)dz).
\end{align}
When $t < l_1$, since the cumulative distribution function $F$ is non-decreasing, we have:
$$
\frac{l_1}{\calT} \cdot (u_1^s(c(l_1)) - u_1^s(c(t))) = \frac{\mathbb{E}_{\vec v}[V(\vec v)]}{\calT} (\int_{t}^{l_1} F^{n-1}(z)dz - F^{n-1}(t) \cdot (l_1-t)) \geq 0.
$$
Similarly, when $t > l_1$, we find:
$$
\frac{l_1}{\calT} \cdot (u_1^s(c(l_1)) - u_1^s(c(t))) = \frac{\mathbb{E}_{\vec v}[V(\vec v)]}{\calT} (F^{n-1}(t) \cdot (t-l_1) - \int_{l_1}^{t} F^{n-1}(z)dz) \geq 0.
$$
Therefore $c(l_1)$ maximizes solver 1's utility, and thus, $c(l_1)$ is the best response.
Consequently, the strategy profile defined by $c(l)$ constitutes the unique symmetric BNE in this mechanism.
\end{proof}

\begin{proposition}[Solver BIR]
\label{prop:solver-bir}
In the equilibrium of Theorem~\ref{thm:bne}, the interim expected utility of a solver with discount ability $l$ is $\frac{\mathbb{E}_{\vec v}[V(\vec v)]}{l}\int_{1}^{l}x(z)\,dz \ge 0$; in particular, the mechanism is BIR for solvers.
\end{proposition}
\begin{proof}
By the equilibrium strategy of Theorem~\ref{thm:bne}, $\frac{c(l)\calT}{l} = \frac{\mathbb{E}_{\vec v}[V(\vec v)]}{l}\left(l\,x(l)-\int_{1}^{l}x(z)\,dz\right)$, so the interim utility is
\[
\mathbb{E}_{\vec v}[V(\vec v)]\,x(l) - \frac{c(l)\calT}{l} = \frac{\mathbb{E}_{\vec v}[V(\vec v)]}{l}\int_{1}^{l}x(z)\,dz \ge 0.
\]
\end{proof}

\subsection{Proof of Theorem~\ref{thm:sym}}
\label{proof:thm-sym}
\begin{proof}
    Consider any strategy profile $\vec{s}$ that maps each solver's discount ability to a cost.
    Since the payment rule is payment-proportional and rank-based, if there exist two solvers, $j$ and $j'$, with different strategies, the expected payment that solver $j$ receives when it either wins, loses, or ties with solver $j'$ can be decomposed into three components:
    \begin{align*}
        \alpha_{W}(c_j) &= \mathbb{E}_{\vec{v}, \vec{l}_{-j}}[V(\vec{v}) \cdot x_j(c_j, \vec{s}_{-j}(\vec{l}_{-j})) ~|~ c_j > s_{j'}(l_{j'})] & (\text{win})\\
        \alpha_{L}(c_j) &= \mathbb{E}_{\vec{v}, \vec{l}_{-j}}[V(\vec{v}) \cdot x_j(c_j, \vec{s}_{-j}(\vec{l}_{-j})) ~|~ c_j < s_{j'}(l_{j'})] & (\text{lose})\\
        \alpha_{T}(c_j) &= \mathbb{E}_{\vec{v}, \vec{l}_{-j}}[V(\vec{v}) \cdot x_j(c_j, \vec{s}_{-j}(\vec{l}_{-j})) ~|~ c_j = s_{j'}(l_{j'})]
        & (\text{tie})
    \end{align*}
    We now reduce to a standard all-pay auction. Multiplying the utility of solver $j$ by the positive constant $l_j/\calT$, which preserves its preference over costs, yields
    \[
        \frac{l_j}{\calT}\, u^s_j \;=\; v_j \cdot x_j(\vec c) - c_j, \qquad \text{where } v_j = \frac{\mathbb{E}_{\vec v}[V(\vec v)]\, l_j}{\calT},
    \]
    and $u^s_j$ denotes solver $j$'s interim expected utility. This is exactly the utility of a bidder with value $v_j$ in an all-pay auction with rank-based allocation who bids $c_j$ and pays her bid, where bids below the due-time cutoff $c(\psi_n^{-1}(0))$ never win, i.e., a reserve in bid space. The allocation remains rank-based in the bids, and the values $v_j$ are i.i.d.\ from a continuous bounded-support distribution (a rescaling of $F$, which admits a density); the two games therefore have identical best responses for every type and hence identical sets of BNE. Theorem 3.1 and Corollary 3.3 of Chawla et al.~\cite{chawla2013auctions} show that all BNE of such auctions are symmetric; hence every BNE of our competition is symmetric, and by Theorem~\ref{thm:bne} it must coincide with $c(l)$.
    %By a similar procedure to Theorem 3.1 from Chawla et al.~\cite{chawla2013auctions},
    % We can outline a proof sketch to rule out asymmetric equilibria(details can be found in the original paper):
    % \begin{itemize}
    %     \item First, reduce the question of the existence of asymmetric equilibria in n-player auctions to that in two-agent auctions, where there exist two player with different strategies;
    %     \item Next, examine the relation between the two players' strategies and the winning, losing and ties probability defined by our new interim allocation rule;
    %     \item If one agent's strategy consistently outbids another's over an interval, then there exists an extended interval where their utilities cross; but if two strategies satisfy utility crossing at some points, they must be equal at all other points in that interval, leading to the conclusion that the strategies must be identical everywhere except on a measure-zero set.
    %     Thus contradiction arises and the result follows.
    % \end{itemize}
\end{proof}

\subsection{Proof for Theorem~\ref{thm:hbw}}
\label{proof:thm-hbw}
\begin{proof}
% For risk-neutral bidders with valuations drawn i.i.d. from F, Myerson[1981] characterized interim allocation and payment rules that arise in BNE:
% \begin{itemize}
% \item[1.]The interim allocation rule x(v) for each agent is monotone non-decreasing in v;
% \item[2.](payment identity) The interim payment rule satisfies $p(v)= vx(v)-\int_{\underline v}^vx(z)dz$
% \end{itemize}
% In all pay auction, payment is equal to the bid. Thus the corresponding bidding function in BNE is exactly expressed by expected payment formula.
% We will rely on a generalized version of Myerson theorem characterizing auctions in BNE.
% \begin{lemma}
%     For each BIC,BIR auction with allocation rule $x(v)$ with  randomness $\omega$ besides others' type affecting the number of 'items', the interim payment rules still satisfy $p(v)=vx(v)-\int_{\underline v}^vx(z)dz$.
% \end{lemma}
% We left proof of the lemma to the Appendix.

% For V-proportional rank-based interim allocation (with a reserve price)$x(v)=E[V]x_0(v)$, we also have the payment or bidding function to be symmetrically $p(v)=E[V]*(vx(v)-\int_{\underline v}^vx(z)dz):=E[V]*p_0(v)$. 

For any rank-based payment-proportional solver payment rule $p_j^s(\vec b, \vec c) = V(\vec b) \cdot x_j(\vec c)$, we define the interim payment allocation function as $\overline{x}_j(l_j) = \mathbb{E}_{\vec l_{-j} \sim F}[x_j(\vec c(\vec l))|l_j]$.
According to the proofs of Theorem~\ref{thm:bne} and Theorem~\ref{thm:sym}, the solver's strategy in the unique BNE is given by:
$$
c(l) = \frac{\mathbb{E}_{\vec{v}}[V(\vec{v})]}{\calT} \cdot (l\overline{x}(l)-\int_{1}^l \overline{x}(z)dz)
$$.

Assuming that $\vec v$ and $\vec l$ are ordered in descending order for the convenience of notation (e.g., $v_{(1)}$ represents the highest bidder value, and $l_{(1)}$ represents the solver's highest discount ability), the expected revenue of the auctioneer is expressed as:
\begin{align}
R &= \mathbb{E}_{\vec v,\vec l}[(v_{(2)} -V(\vec v)) \cdot \kappa(c(l_{(1)}))]/\calT \nonumber\\
&= \mathbb{E}_{\vec v}[v_{(2)} -V(\vec v)] \cdot \mathbb{E}_{\vec v, \vec l}[\kappa(c(l_{(1)}))]/\calT \nonumber\\
&=
\mathbb{E}_{\vec v}[v_{(2)} -V(\vec v)] \cdot E_{\vec l}\left[\kappa\left(\frac{\mathbb{E}_{\vec {v}}[V(\vec v)]}{\calT}(l_{(1)}\overline x(l_{(1)})-\int_{1}^{l_{(1)}} \overline x(z)dz)\right)\right] \nonumber \\
&= \mathbb{E}_{\vec v}\left[v_{(2)} - V(\vec v)\right]\left(\beta \frac{\mathbb{E}_{\vec {v}}[V(\vec v)]}{\calT} 
\mathbb{E}_{\vec l}\left[l_{(1)}\overline x(l_{(1)})-\int_{1}^{l_{(1)}} \overline x(z)dz\right]+\gamma\right)/\calT \label{rev}
\end{align}

The speed factor above is taken on the event that an eligible solution exists, i.e., $l_{(1)} \ge \psi_n^{-1}(0)$; on the complementary event no submission arrives before $t_d$, no solver payment is made, and the auction finalizes from valid self-openings, so the auctioneer's revenue on that event does not depend on $V(\vec v)$ and is irrelevant to the optimization below.
We set $\mathbf{MP} = \mathbb{E}_{\vec l}[l_{(1)}\overline x(l_{(1)})-\int_{1}^{l_{(1)}} \overline x(z)dz]$.
Since $\mathbf{MP}$ is independent of $\vec{v}$, we first characterize the optimal choice of $\overline{x}(l)$ that maximizes this function.

Given that $l \overline x(l)-\int_{1}^l \overline x(z)dz$ is monotone in $l$, the solver with the highest discount ability will always be the winner.
Thus, the expected maximum value of $\mathbf{MP}$ is:

\begin{align*}
% \mathbf{MP}_{i}[v_ic_i\text{ in }\mathcal{A}] &=\int_{\underline{v}}^{\bar{v}}b(v_{i})\mathbf{Pr}_{\mathbf{v}_{-i}}\begin{bmatrix}v_{i}=v_{(1)}\end{bmatrix}f(v_{i})dv_{i}  \\
%&=n\mathbf{MP}_{j}[\mathcal{A}] \text{(assuming contributed by solver $j$) }\\
\mathbf{MP}
&=n\int_{1}^{\bar{l}}\left[l\overline x(l)-\int_{1}^{l}\overline x(z)dz\right]\mathbf{Pr}_{\mathbf{l}_{-j}} \begin{bmatrix}l=l_{(1)}\end{bmatrix}f(l)d l  \\
&=n\int_{1}^{\bar{l}}\left[l\overline x(l)-\int_{1}^{l}\overline x(z)dz\right] F(l)^{n-1}f(l)d l.\\
&(\text{interchange the order of integration over $z$ and $l$})\\
&=n\int_{1}^{\bar{l}}l \overline x(l) F(l)^{n-1} f(l)d l-\int_{1}^{\bar{l}}\overline x(l)\left(1-F(l)^{n}\right)dl  \\
&=n\int_{1}^{\bar{l}}\left(lF(l)^{n-1}-\frac{1-F(l)^{n}}{n f(l)}\right)\times \overline x(l)f(l)dl \\
&=n\int_{1}^{\bar{l}}\psi_{n}(l)\overline x(l)f(l)dl \\
&=n\mathbf{E}_{l\sim F}[\psi_{n}(l)\overline x(l)].
% &=\int_{\underline{v}}^{\bar{v}}b(v_{i})d(F(v_{i})^n)\\
% (&\text{If } g(x)=kx+b(k>0))\\
% \mathbf{MP}[\mathcal{A}]
% &=kn\int_{\underline{v}}^{\bar{v}}b(v_i)\textbf{Pr}_{\mathbf{v}_{-i}}[v_i=v_{(1)}]f(v_i)dv_i+b \\
% &=kn\mathbf{MP}_{i}[v_ic_i\text{ in }\mathcal{A}]+b\\
% &=kn\mathbb E_{v\sim F}[\psi_n(v)x(v)]+b \\
% &(\text{If $g$ is convex})\\
% \mathbf{MP}[\mathcal{A}]
% &\geq g(\int_{\underline{v}}^{\bar{v}}b(v_{i})d(F(v_{i})^n))\\
% &=g(n\mathbf{MP}_{i}[v_ic_i\text{ in }\mathcal{A}])\\
% &=g(n\mathbf{E}_{v_{i}\sim F}[\psi_{n}(v_{i})x(v_{i})])
\end{align*}

This result shows that when the entire payment is allocated to the solver with the largest (non-negative) virtual value, $\mathbf{MP}$ is maximized. 
For a $n$-maximum-payment regular distribution $F$, the virtual value function is monotonically increasing, implying that winner-takes-all competition is optimal.
This result also implies that the optimal due time is $t_d = \frac{\mathcal{T}}{\kappa(c(\psi_n^{-1}(0)))}$.

Next, we calculate the optimal $V(\vec{v})$ that maximizes the auctioneer's revenue in a winner-takes-all competition.
Specifically, we need to maximize equation~\eqref{rev} with respect to $V(\vec{v})$. 
Using the Arithmetic Mean-Geometric Mean (AM-GM) inequalities, as long as $\gamma$ is not too large (i.e., $\gamma \leq \mathbb{E}_{\vec v}[v_{(2)}] \cdot \beta \cdot \mathbf{MP} / \calT$), the optimal function $V$ satisfies:
$$
\mathbb{E}_{\vec v}[V(\vec{v})] = \frac{\mathbb{E}_{\vec v}[v_{(2)}]}{2} - \frac{\gamma \cdot \calT}{2n \beta \mathbb{E}_{l \sim F}[\psi_n(l)\overline x(l)]}\ge 0.
$$
Further, to achieve IR for the auctioneer(namely, make the revenue nonnegative for every $\vec{v},\vec l$), just let 
% $E_{v_{-(2)}}V(\vec{v}):=(v_{(2)}-\frac{\gamma}{\beta\mathbf{MP}})/2$ or more directly 
$$
V(\vec{v}) = \frac{v_{(2)}}{2} - \frac{\gamma \cdot \calT \cdot v_{(2)}}{2n \beta \mathbb{E}_{l \sim F}[\psi_n(l)\overline x(l)]\cdot \mathbb{E}_{\vec v}[v_{(2)}]}.
$$

%We mention that this result can also extend to other bidder's mechanism by change the $v_{(2)}$ to auctioneer exact revenue.

% $$V(\vec{v}):=(v_{(2)}+\frac{\gamma*v_{(2)}}{\beta\mathbf{MP}[\mathcal{A}]*E_{v_{(2)}}[v_{(2)}]})/2$$,

%Also, this value $E_{\vec{v}}[V]$ can be precomputed and published by auctioneer to save servers' effort.

%The resulting optimal auction turns out to have a symmetric BNE (see Proof of Theorem \ref{thm:bne}), thus reaching exact optimality.
\end{proof}

% \textbf{Remark. }In other cases where the value distribution $F$ is limited to the distribution set with zero support at values smaller than $\underline{V}$, the expected revenue of a two-side DSIC mechanism cannot reach a positive fraction of the revenue of our mechanism. That is, our the ratio of its revenue and the revenue of our mechanism can approach to infinity. For example, let $g(x)=x$ and the value distribution is $U[\underline{V},\overline {V}]$. Then for any two-side DSIC mechanism, the computing cost shouldn't exceed $\underline{V}$. 

% Consider the distribution $U[\frac{\underline{V}} k,\frac{\overline {V}} k]$. It's easy to observe that the ratio is not sensitive to the positive scaling $k$. Let $\overline V=k$ grows to infinity, our revenue is approximately the one running on distribution $U[0,1]$, which is strictly positive. while the above discussion shows the revenue(certainly smaller than welfare) of this two-sides DSIC mechanism approach to zero.

% The ratio will be no larger than 
% \begin{align}
% &(\text{Denote $\max_i{v_i}$ as $v_{max}$})\\
% &\quad\frac
% {E_{\vec {v},\vec l}[v_{max}*(\underline{v}*\max_jl_j)]}{rev}
% \\&= \frac{E_{\vec {v},\vec l}[v_{max}^2]}
% {E_{\vec {v},\vec l}[v_{max}*\underline{v}]}
% \\&\ge \frac{E_{\vec {v},\vec l}[v_{max}]}
% {\underline{v}}
% \ge \frac{\underline{v}+\overline{v}}{2\underline{v}}
% \end{align}

\begin{corollary}
\label{cor:nonlinear}
    If $\kappa$ is an increasing function, under the winner-takes-all solver payment rule, an optimal $V$ is decided by the following optimization problem:
    \begin{align}
    \max_{V} \mathbb{E}_{\vec{v}}[v_{(2)}-V(\vec{v})] \cdot \mathbb{E}_{l \sim F}[\kappa(\frac{\mathbb{E}_{\vec{v}}[V(\vec{v})]}{\calT} \cdot (l_{(1)}x(l_{(1)})-\int_{1}^{l_{(1)}} x(z)dz))],
    \end{align}
    where $v_{(2)}$ is the second-highest value among the bidders (or the reserve price), and $l_{(1)}$ is the highest discount ability among the solvers.
\end{corollary}

\subsection{Proof of Theorem~\ref{thm:two-side-DSIC}}
\label{proof:thm-two-side}
\begin{proof}
We prove it by contradiction. 
Suppose mechanism $M$ is two-sided DSIC and ex-post IR, satisfies ex-post auctioneer IR, and uses a positive-cost opening action with positive probability. 
Let that action have marginal computational cost $c>0$.
Note that every opening action that can complete has positive cost: a solver running at cost $0$ has speed $\kappa(0)=0$ and never finishes, and by forced opening the bid vector cannot be revealed through bidder-cooperative messages, nor by the auctioneer's own computation, which the mechanism class excludes.

Before the aggregate is opened, cryptographic hiding implies that a solver's information does not reveal the bid vector. Thus the same positive-cost action is available, and is selected with the same probability, on every value profile consistent with the solver's observations. Consider a profile $\vec v$ such that
\[
    \sum_{i=1}^m v_i<c,
\]
which exists when bidder values can be zero.

On this realized profile, ex-post bidder IR and allocation feasibility imply
\[
    \sum_{i=1}^m p_i^b(\vec v)\le \sum_{i=1}^m v_i<c.
\]
Ex-post solver IR requires the solver payment on the opening event to be at least $c$. Therefore the auctioneer's realized revenue is
\[
    \sum_{i=1}^m p_i^b(\vec v)-\sum_{j=1}^n p_j^s(\vec v,\vec c)<0,
\]
contradicting ex-post auctioneer IR. Hence no positive-cost opening action can be used on such a profile. Since bidder values can be zero, such a profile exists for every positive-cost opening action, so no positive-cost opening action is selected with positive probability. Without an opening action the mechanism can never determine the bids; by bidder IR it then collects no payments, so its revenue is zero on every value profile.
\end{proof}

\section{Related Work}
\label{appendix:related}

\noindent\textbf{Blockchain Sealed-bid Auction.}
% Seal-bid auction is a widespread, important economic process. Vickery~\cite{vickrey1961counterspeculation} and Myerson~\cite{myerson1981optimal} established the standard auction theory. 
% The Vickery auction is simple and incentive-secure, while Myerson~\cite{myerson1981optimal} focused on revenue-maximizing auctions. The revenue equivalence result shows that two auction mechanisms with a Bayesian Nash Equilibrium (BNE) and the same allocation yield the same expected payments, reducing the focus on payment rules in auction design.
%Previous works proposed various approaches to bid privacy to disincentivize aborting. 
Several papers have proposed sealed-bid auction blockchain protocols. Some focus on bid privacy but do not ensure auction completion~\cite{braghin2020designing,chen2022sbrac,krol2020pastrami,lu2021blockchain,blass2018strain}. 
Block Auction uses a smart-contract orchestrated cryptographic protocol to support verifiable periodic double auctions without a single trusted auctioneer~\cite{constantinides2021block}. Other completion-oriented designs use different trust assumptions, including off-chain managers~\cite{desai2019hybrid,galal2019verifiable,sharma2021anonymous} or trusted execution environments~\cite{desai2021secauctee,galal2020trustee}. To address both bid privacy and auctioneer credibility, Tyagi et al.~\cite{tyagi2023riggs} proposed the Riggs protocol, which utilizes a non-malleable timed commitment scheme with range proofs, enabling penalties for abandoned bids and allows simultaneous bids in multiple auctions. Cicada~\cite{glaeser2023cicada} also offers a comprehensive private on-chain auction framework, using aggregate homomorphic time-lock puzzles to protect bidder privacy. 
Recent systems further study delayed execution without deposits, private auctions at scale with publicly verifiable results, liveness in rational settings, censorship resistance with simultaneous bid release, and TEE-backed confidential compute~\cite{zhang2024zeroauction,novakovic2026cryptobazaar,huang2025optimizing,alpos2026censorship,gebele2025sealed}.
Sealed-bid auctions have also been studied on consortium blockchain~\cite{xiong2019anonymous} and quantum-based blockchain~\cite{abulkasim2021quantum}.

\smallskip
\noindent\textbf{Timed Commitment.}
Time-lock puzzles were first proposed by Rivest et al.~\cite{rivest1996timelock}, using an RSA-based assumption to create puzzles that can be opened through repeated squaring. Timed commitments~\cite{boneh2000timed} are a related but distinct primitive.
Various time-lock puzzles with different properties have since been introduced, such as homomorphic time-lock puzzles~\cite{malavolta2019homomorphic}, fully homomorphic time-lock puzzles~\cite{brakerski2019leveraging}, non-malleable time-lock puzzles~\cite{freitag2021nonmalleadle}, sequential time-lock puzzles~\cite{chvojka2021versatile}, and the ``Efficient Delegated Time-Lock Puzzle'' (ED-TLP)~\cite{abadi2023delegated}. 
Scalable time-lock puzzles support aggregation, delegated solving, variable delays, and fair payments~\cite{abadi2025scalable}, while verifiable time-lock puzzles additionally check puzzle validity before solution and verify solution correctness afterwards~\cite{xin2025check,wu2024practical}.
The hardness of the underlying assumption, the sequential-squaring conjecture, is studied in~\cite{katz2020security}, and the applications of timed commitment are discussed in~\cite{boneh2000timed,thyagarajan2020verifiable,choi2023bicorn}. 
Furthermore, techniques like group homomorphism in hidden-order groups~\cite{bangerter2005efficient,boneh2019batching,castagnos2020bandwidth} have been employed on non-malleable timed commitments, which efficiently construct range proof for committed value~\cite{thyagarajan2021efficient}.

\smallskip
\noindent\textbf{Outsourcing of Computation Task.}
The outsourcing of computation tasks has been extensively studied, but the focus has primarily been on the security and efficiency of outsourcing, with little attention given to the issue of outsourcing pricing~\cite{gennaro2010non,wang2012harnessing,shan2018practical}. 
The emergence of blockchain has enabled decentralized computation outsourcing. 
In 2021, Thyagarajan et al.~\cite{thyagarajan2021opensquare} introduced OpenSquare, a blockchain-based smart contract that allows clients to outsource repeated modular squaring computations to computationally powerful servers. 
This system is designed to be Sybil-resistant and proven incentive-compatible through game-theoretic analysis. 
However, some studies suggest that OpenSquare may be vulnerable to collusion attacks~\cite{abadi2023decentralised}. 
Another type of computation task that may require outsourcing is the generation of zero-knowledge (ZK) proofs, which is explored in the mechanism proposed by Wang et al.~\cite{wang2024mechanism}.
Gong et al.~\cite{gong2024v3rified} study deadline-aware solver competition for decentralized verifiable computation and compare revelation with non-revelation mechanisms.
Blockchain procurement auctions have also been used to analyze adversarial liveness risks and alternatives to winner-take-all allocation, including backup solvers, slashing, and multiple procured executors~\cite{bahrani2026adversarial,garimidi2026beyond}.

\smallskip
\noindent\textbf{Bridging Game Theory and Cryptography.}
The intersection of game theory and cryptography aims to develop more realistic security models by incorporating participant incentives and computational constraints. 
Asharov et al.~\cite{asharov2011towards} formalized traditional cryptographic security concepts like secrecy, correctness, and relaxed fairness as properties of computational Nash equilibria, linking the concept in cryptography and game theory. 
Building on this, \cite{garay2013rational} introduced a meta-game framework, rational protocol design, defining attacker utility via simulation and corruption costs, which enables composability and bypasses impossibility results in secure function evaluation. 
In auction settings, Chung et al.~\cite{chung2025foundations} model platform-assisted auctions with strategic platforms and coalitions and use cryptographic tools to obtain mechanism-design guarantees.
On the other side, Halpern and Pass~\cite{halpern2015algorithm} laid the groundwork for game theory with costly computation, demonstrating how computational costs influence equilibrium existence and explain bounded rationality. 
For practical protocols, \cite{abraham2006distributed} contributed to this intersection by developing k-resilient mechanisms for secret sharing and multi-party computation, achieving robustness against coalitions. 
These works refine traditional security definitions to account for practical limitations and strategic behavior, paving the way for designing protocols that are both secure and incentive-compatible.

\end{document}